%% file: main.tex
\documentclass[11pt]{article}
\usepackage[margin=1in]{geometry}
\usepackage{amsmath,amsthm,amsfonts,amssymb}
\usepackage[colorlinks=true,linkcolor=magenta!70!black,citecolor=blue!60!black,urlcolor=blue!60!black]{hyperref}
\usepackage{enumitem}
\usepackage[T1]{fontenc}
\usepackage{aliascnt}
\usepackage{graphicx}
\usepackage{tikz}
\usetikzlibrary{arrows.meta,calc,positioning}
\usepackage{float}
\usepackage{xcolor}
\usetikzlibrary{arrows.meta}

\definecolor{slotred}{RGB}{205,55,55}
\definecolor{slotblue}{RGB}{45,105,190}
\definecolor{slotgreen}{RGB}{45,145,75}
\definecolor{slotyellow}{RGB}{218,167,0}

\usepackage{booktabs,tabularx,array}
\usepackage{caption}
\usepackage{bbm}

\allowdisplaybreaks

\usepackage[capitalise,nameinlink]{cleveref}
\Crefname{theorem}{Theorem}{Theorems}
\Crefname{subsection}{Subsection}{Subsections}
\Crefname{section}{Section}{Sections}
\Crefformat{equation}{(#2#1#3)}
\crefname{subsection}{Subsec.}{Subsecs.}
\Crefname{subsection}{Subsec.}{Subsecs.}

\newtheorem{theorem}{Theorem}[section]

\newcommand{\newaliastheorem}[2]{%
  \newaliascnt{#1}{theorem}%
  \newtheorem{#1}[#1]{#2}%
  \aliascntresetthe{#1}%
  \crefname{#1}{#2}{#2s}%
  \Crefname{#1}{#2}{#2s}%
}

\newaliastheorem{proposition}{Proposition}
\newaliastheorem{lemma}{Lemma}
\newaliastheorem{conjecture}{Conjecture}

\theoremstyle{definition}
\newaliastheorem{definition}{Definition}
\newaliastheorem{remark}{Remark}
\newaliastheorem{corollary}{Corollary}

\usepackage{thmtools}
\usepackage{thm-restate}

\newcommand{\floorbra}[1]{\left\lfloor #1 \right\rfloor}

\newcommand{\classp}{\mathsf{P}}
\newcommand{\classnp}{\mathsf{NP}}
\newcommand{\dtime}{\textsc{DTIME}}

\title{Inapproximability of Unique-Machine Precedence Scheduling for Unit-Length Jobs}

\author{
Venkatesan Guruswami
\thanks{Simons Institute for the Theory of Computing, and Departments of EECS \& Mathematics, UC Berkeley. Email: {\tt venkatg@berkeley.edu}. Research supported in part by NSF grant CCF-2211972 and a Simons Investigator award.}
\and 
Xuandi Ren\thanks{Department of EECS, UC Berkeley. Email: \texttt{xuandi\_ren@berkeley.edu}. Supported in part by NSF grant CCF-2228287, a DARPA grant under Contract No. HR0011262E031, and V.G's Simons Investigator award.}
\and 
Shaoxuan Tang\thanks{Institute for Interdisciplinary Information Sciences, Tsinghua University. This work was done when visiting UC Berkeley. Email: \texttt{tsx23@mails.tsinghua.edu.cn}.}
}

\date{}
\begin{document}

\maketitle
\thispagestyle{empty}

\begin{abstract}
The Unique-Machine Precedence Scheduling (UMPS) problem, introduced by (Davies, Kulkarni, Rothvoss, Sandeep, Tarnawski, and Zhang, 2022), seeks a makespan-minimizing schedule of precedence-constrained jobs when each job has a unique eligible machine. On the one hand, UMPS generalizes job shop scheduling by allowing the precedence graph to be an arbitrary DAG rather than a disjoint union of chains. On the other hand, UMPS admits approximation-preserving reductions to scheduling problems with communication delays, including the job-job delay model (DKRSTZ, 2022) and the job-machine delay model (Rajaraman, Stalfa, and Yang, 2023). Despite its central role, the approximability of UMPS has remained poorly understood: even for unit-length jobs, known scheduling techniques do not seem to yield a non-trivial approximation, and the existence of a polylogarithmic approximation was left open by (DKRSTZ, 2022). 
On the hardness side, the previous best lower bound for unit-length jobs was only the 5/4 inherited from job shop scheduling (Williamson, Hall, Hoogeveen, Hurkens, Lenstra, Sevast'janov, and Shmoys, 1997). 

\smallskip
We prove that unit-length UMPS is NP-hard to approximate within any constant factor. We further show that, assuming NP is not in quasi-polynomial time, unit-length UMPS admits no polynomial-time $(\text{log} n)^\gamma$-approximation for some constant $\gamma>0$. Via the known reductions from UMPS, these lower bounds also transfer to the corresponding unit-length communication-delay scheduling models.

\smallskip
Our proof proceeds via a reduction from a hypergraph coloring promise problem. In the yes case, the input $h$-uniform hypergraph admits a $k$-coloring in which no color appears more than $a$ times in any hyperedge; in the no case, the hypergraph has no large independent set. We convert such a promise into a UMPS gap of order $h/a$.  
Instantiating this framework with a balanced hypergraph coloring hardness of (Guruswami and Lee, 2018) gives arbitrary constant-factor inapproximability, while combining the $4$-colorable $4$-uniform hypergraph coloring hardness of (Guruswami, Harsha, H{\aa}stad, Srinivasan, and Varma, 2017) with a certain composition operation for  hypergraphs yields the polylogarithmic factor inapproximability.
\end{abstract}

\newpage
\tableofcontents
\thispagestyle{empty}

\clearpage
\input{sec/intro}
\input{sec/pre}
\input{sec/reduction}
\input{sec/coloring}

\section*{Acknowledgements}
The authors used ChatGPT 5.5 Pro to assist in making the figures in this paper.

\nocite{*}
\bibliographystyle{alpha}
\bibliography{ref}

\clearpage

\appendix
\input{sec/pcp-free}

\end{document}

%% file: sec/intro.tex
\section{Introduction}
Precedence constraints are a classic source of computational difficulty in scheduling problems~\cite{LenstraRinnooyKan1978Precedence,GrahamEtAl1979Survey}. Jobs form a directed acyclic graph: a job can start only after all its predecessors have completed, and the makespan objective asks to complete the whole DAG as quickly as possible.  This basic model has two intertwined sources of difficulty: the partial order among jobs and the constraints governing machine use. 
Several classical variants enrich the latter: 
machine-eligibility constraints restrict each job to a prescribed subset of machines~\cite{LeungLi2008ProcessingSet,LeeLeungPinedo2011MachineEligibility}; %
unrelated-machine models let the processing time $p_{ij}$ depend on both the job and the machine~\cite{LenstraShmoysTardos1990Unrelated}; and communication-delay models impose extra separation when dependent jobs are executed on different machines~\cite{PapadimitriouYannakakis1990ArchitectureIndependent}.

\emph{Unique-Machine Precedence Scheduling} (UMPS) takes a complementary minimalist view of the job-machine interaction: every job is assigned to a specific machine in advance, so the scheduler can only choose the order of the jobs on each machine, subject to arbitrary  precedence constraints between jobs on possibly different machines. Thus UMPS asks whether precedence scheduling is already hard before one uses assignment choices, machine-dependent processing times, or communication delays.

UMPS lies between several better-known scheduling problems. On the upstream side, job shop scheduling is a special case of UMPS, in which the precedence graph is a disjoint union of chains. Therefore, UMPS inherits the hardness for job shop: quasi-polynomial-time logarithmic inapproximability with arbitrary job lengths \cite{MastrolilliSvensson2011JobShop}, and $5/4$ NP-hardness when jobs have unit lengths \cite{WHH+97}. On the downstream side, in their paper introducing UMPS, Davies, Kulkarni, Rothvoss, Sandeep, Tarnawski, and Zhang~\cite{DaviesEtAl2022NonUniform} reduce it to the non-uniform job-job communication-delay problem, where jobs are assigned to identical machines, and each precedence edge $j\prec k$ has a delay $c_{j,k}$ that must elapse between the completion of $j$ and the start of $k$ if the two jobs are scheduled on different machines. The resulting instance has only unit-length jobs if UMPS is of unit-length. Rajaraman, Stalfa, and Yang~\cite{RajaramanStalfaYang2023JobMachineDelays}
later considered a job-machine delay model, in which the cross-machine delay is a function
\(\rho_{k,i}\) of the successor job \(k\) and the machine \(i\):
if \(k\) starts on machine \(i\), then every predecessor \(j\prec k\) scheduled on a different
machine must complete at least \(\rho_{k,i}\) time before \(k\) starts.  They gave an approximation-preserving reduction from UMPS to this job-machine delay model. 

The algorithmic picture of UMPS remains poorly understood. When the fixed-machine restriction is absent, there are various approximation algorithms for scheduling with precedence, e.g., Graham's list-scheduling 2-approximation on identical machines \cite{Graham1966}, and $O(\log m/\log \log m)$ approximations for related machines \cite{ChudakShmoys1999, ChekuriBender2001, Li2017}. %
For precedence scheduling with communication delays, positive results are known mainly under strong structural restrictions: uniform delays admit polylogarithmic approximations via Sherali--Adams clustering and duplication-based rounding~\cite{DaviesEtAl2020LPHierarchies,MaitiEtAl2020RelatedMachinesCommunication}, while certain additive non-uniform delay models admit polylogarithmic guarantees~\cite{RajaramanStalfaYang2023JobMachineDelays}. These algorithms, however, rely on the flexibility of assigning jobs to suitable machines, which UMPS does not offer.  Therefore, they do not yield comparable guarantees for UMPS itself.  Indeed, the work introducing UMPS explicitly points out that classical job-shop algorithms and recent LP-hierarchy methods fail to provide non-trivial approximation guarantees for UMPS~\cite{DaviesEtAl2022NonUniform}. This led them to formulate the following strong hardness conjecture, even for the unit-length case:

\begin{conjecture}[\cite{DaviesEtAl2022NonUniform}, Conjecture~2]
\label{conj:dkrstz-polynomial}
There exists an absolute constant $\varepsilon > 0$ such that it is $\classnp$-hard to approximate UMPS with $n$ jobs within a factor of $n^\varepsilon$, even when all jobs have unit length.
\end{conjecture}

In this work, we prove that unit-length UMPS is NP-hard to approximate within any constant factor. Moreover, assuming $\classnp\nsubseteq \dtime(n^{\mathrm{polylog}(n)})$, it does not admit a polynomial-time $(\log N)^{\gamma}$-approximation for some constant $\gamma>0$, where $N$ denotes the instance size. These results substantially strengthen the previous $5/4$ inapproximability for UMPS inherited from job shop scheduling, and mark a step toward proving \Cref{conj:dkrstz-polynomial}.

\begin{restatable}{theorem}{thmnp}
\label{thm:np}
    For any constant $\gamma>1$, $\gamma$-approximating unit-length UMPS is NP-hard.
\end{restatable}

\begin{restatable}{theorem}{thmquasi}
\label{thm:quasi}
    Assume $\classnp\nsubseteq \dtime(n^{\mathrm{polylog}(n)})$. For some constant $\gamma>0$, unit-length UMPS has no polynomial-time $(\log N)^\gamma$-approximation, where $N$ is the instance size.
\end{restatable}

Our proof proceeds via reduction from a hard form of hypergraph coloring. Given an $h$-uniform hypergraph $G$, we rely on the following hardness promise: in the yes case, $G$ admits a $k$-coloring in which no color appears more than $a<h$ times in any hyperedge; in the no case, no independent set contains more than an $\varepsilon<1/k^2$ fraction of vertices. By carefully designing vertex and hyperedge machine gadgets to encode the coloring, we show that any hardness of this form yields an inapproximability factor of $O(h/a)$ for unit-length UMPS. We then instantiate our reduction with the hardness result in \cite{Guruswami2018}, where $h/a$ can be made an arbitrarily large constant, thereby proving \Cref{thm:np}. We also analyze how the hypergraph composition method in \cite{ErdosLovasz1975} affects the parameters relevant to our reduction, and thus amplify the gap to $(\log N)^{\Omega(1)}$, proving \Cref{thm:quasi}.

As corollaries, the reductions in \cite{DaviesEtAl2022NonUniform} and \cite{RajaramanStalfaYang2023JobMachineDelays} transfer our lower bounds to two non-uniform communication-delay scheduling problems on identical machines with unit-length jobs: the job-job delay model, where a precedence edge $j \prec k$ has a delay $c_{j,k}$ incurred when $j,k$ run on different machines, and the job-machine delay model, where if job $k$ is scheduled on machine $i$, every predecessor of $k$ scheduled on different machines must finish at least $\rho_{k,i}$ time before $k$ starts. For each of these two problems, there is no polynomial-time constant-factor approximation unless $\classp  = \classnp$; and under $\classnp\nsubseteq \dtime(n^{\mathrm{polylog}(n)})$, there is no polynomial-time $(\log N)^{\gamma}$ approximation for some constant $\gamma>0$, where $N$ is the instance size.

Alongside our main construction, we give a simple PCP-free reduction showing that unit-length UMPS is NP-hard to approximate within a factor of 4/3. This already improves the previous 5/4 inapproximability inherited from unit-length job shop scheduling, and may be of independent interest.  We include it in \Cref{sec:pcp-free}.  %
We also found a simple and clean reduction from $(2+\varepsilon)$-SAT \cite{AustrinGuruswamiHastad2017TwoPlusEpsilonSAT} that gives factor 4/3 NP-hardness of unit-length UMPS, albeit  it is not PCP-free. We omit the details of this reduction.

\subsection{Proof overview}
\label{subsec:proof-overview}

We first explain the reduction in a minimal setting where all parameters are concrete: yes instances have a schedule of makespan $12$, whereas no instances have makespan at least $13$. This is not yet a large approximation gap, but it contains the main ideas of the full construction in \Cref{sec:reduction}.

We use the following hypergraph coloring hardness as the starting point: for every fixed $\varepsilon>0$, it is NP-hard to distinguish a $2$-colorable $4$-uniform hypergraph $H=(V,E)$ from one in which every independent set has size less than $\varepsilon |V|$. Here we take $\varepsilon<1/2$, and the parameters are $k=2, h=4, a=3$, since no hyperedge is monochromatic means each color appears on at most 3 vertices of every hyperedge. This hardness follows from the verifier of Guruswami--H\r{a}stad--Sudan, together with the independent-set strengthening of Holmerin; see \cite{GuruswamiHastadSudan2002,Holmerin2002}. In this subsection, we call the two colors $0$ and $1$, and index the time starting at $0$.

We make two virtual copies $H^0$ and $H^1$ of the hypergraph $H$. Suppose in the yes case there is a proper coloring $\chi:V\to\{0,1\}$. We color $H^0$ by $\chi$ and $H^1$ by the complementary coloring $1-\chi$. Thus the two copies of the same vertex receive different colors, while each copy of the hypergraph is still properly colored. For every vertex $v$, we create two vertex machines  $M_{v,\mathrm L}$ and  $M_{v,\mathrm R}$. On each machine, we put two length-$3$ chains, one for each virtual copy $i\in\{0,1\}$:
\[
J_{v,\mathrm L}^{i,1}\prec J_{v,\mathrm L}^{i,2}\prec J_{v,\mathrm L}^{i,3},
\qquad
J_{v,\mathrm R}^{i,1}\prec J_{v,\mathrm R}^{i,2}\prec J_{v,\mathrm R}^{i,3}.
\]
We write $S_{v,i,\mathrm L},T_{v,i,\mathrm L}$ and $S_{v,i,\mathrm R},T_{v,i,\mathrm R}$ for the heads and tails of these chains. Think of each vertex machine as having two consecutive length-$3$ color slots. Placing the chain for virtual copy $i$ in slot $b\in\{0,1\}$ means that the copy $v^{(i)}$ is assigned color $b$. Since the two virtual copies of a vertex have opposite colors in the yes case, the two chains of $v$ occupy the two slots without conflict. 

The edge machines check this coloring one hyperedge copy at a time. For every hyperedge $e\in E$ and virtual copy $i\in\{0,1\}$, create one edge machine $M_{e,i}$. For each vertex $v\in e$, put a unit job $Y^v_{e,i}$ on $M_{e,i}$ and impose the bridge precedence constraints
\[
T_{v,i,\mathrm L}\prec Y^v_{e,i}\prec S_{v,i,\mathrm R}.
\]
Thus the edge job for $v$ must run after the left copy of the corresponding vertex chain has finished and before the right copy of the same chain begins.

Now consider the yes case. If virtual copy $i$ gives $v$ color $b$, we schedule the chain $J^{i,*}_{v,\mathrm L}$ in the slot $[3b,3(b+1))$, and schedule the chain $J^{i,*}_{v,\mathrm R}$ in the slot shifted two slots to the right, namely $[3(b+2),3(b+3))$. The bridge constraints leave the middle window
\[
[3(b+1),3(b+2))
\]
for edge jobs of color $b$. On the edge machine $M_{e,i}$, schedule all jobs $Y^v_{e,i}$ with copy-$i$ color $b$ in this window. Because the copy-$i$ coloring is proper, a $4$-edge has at most three vertices of color $b$, so those unit jobs fit in the length-$3$ window. This gives a feasible schedule of makespan $12$.

For soundness, we prove the makespan is at least 13. Suppose toward contradiction that there is a schedule of makespan less than 13. For each vertex $v$, choose the virtual copy $i(v)\in\{0,1\}$ whose left tail $T_{v,i,\mathrm L}$ completes later. The left machine $M_{v,\mathrm L}$ contains six unit jobs, so the chosen left tail completes at or after time $6$. On the other hand, the right chain for the same virtual copy has length $3$ and finishes before time $13$, so its head $S_{v,i(v),\mathrm R}$ starts before time $10$. Therefore every bridge job $Y^v_{e,i(v)}$ with $e\ni v$ must start in the interval $[6,9)$: it cannot start before $T_{v,i(v),\mathrm L}$ completes, and it must finish before $S_{v,i(v),\mathrm R}$ starts.

By pigeonholing, some virtual copy $i\in\{0,1\}$ is selected by at least $|V|/2>\varepsilon|V|$ vertices. In the no case this selected set is not independent, so it contains a hyperedge $e$. For every $v\in e$, the job $Y^v_{e,i}$ lies on the same edge machine $M_{e,i}$, and all four of these unit jobs must start in $[6,9)$. However, four unit jobs cannot be scheduled non-overlapping in this length-3 window. This contradiction shows that every no-case schedule has makespan at least $13$.

The full proof in \Cref{sec:reduction} uses the same mechanism with larger parameters. Instead of two colors and two virtual copies, it uses $k$ colors and $k$ virtual copies; a balanced coloring is cyclically shifted across the virtual copies so that the chains of each vertex fill all $k$ color slots. Instead of a single left-to-right bridge, it uses $c$ consecutive bridges. In the no case, the proof still only selects, for each vertex, one virtual copy and one bridge where the corresponding chain is forced across a narrow time window. Pigeonholing over the $kc$ choices gives a large selected vertex set, the independent-set promise supplies a hyperedge inside it, and the edge-machine jobs for that hyperedge are squeezed into too little time. We refer the readers to \Cref{fig:vertex-machine}, \Cref{fig:edge-machine-completeness} and \Cref{fig:soundness-illustration} for illustrations of the vertex machine, the completeness case, and the soundness case, respectively. Although drawn for the full construction, they might be also helpful to understand the toy case in this section.

We introduce the problems and the notation in \Cref{sec:pre}, present our main reduction from hypergraph coloring to unit-length UMPS in \Cref{sec:reduction}, and analyze the hypergraph composition in \Cref{sec:hardness-coloring}.

%% file: sec/pre.tex
\section{Preliminaries}
\label{sec:pre}

\begin{definition}[The Unique-Machine Precedence Scheduling (UMPS) Problem]
\label{def:umps}
    An instance of the unique-machine precedence scheduling problem consists of a set of machines $\mathcal M$, a set of jobs $\mathcal J$, a precedence relation $\prec$ forming a directed acyclic graph, a map $\sigma:\mathcal J \to \mathcal M$ assigning every job to its unique eligible machine, and a function $\ell:\mathcal J \to \mathbb R^+$ specifying the processing time of each job. A feasible schedule assigns each job $j$ a starting time $s(j)$ such that
    
    \begin{enumerate}
        \item $s(j)\ge 0$,
        \item jobs assigned to the same machine do not overlap, and
        \item $j \prec j'$ implies $s(j)+\ell(j) \le s(j')$.
    \end{enumerate}

    The objective is to minimize the makespan
    $$C_{\text{max}}=\max_{j \in \mathcal J} \{s(j)+\ell(j)\}.$$
\end{definition}

\begin{definition}[$a$-Bounded $k$-Coloring of Hypergraphs]
\label{def:ka_coloring}
    An $a$-bounded $k$-coloring, abbreviated as an $(a,k)$-coloring, of a hypergraph $G=(V,E)$ is a mapping $\chi:V \to [k]$, such that for every $e \in E$ and $r \in [k]$,
    $$|e \cap \chi^{-1}(r)|\le a.$$
\end{definition}
This notion is also called an $a$-proper $k$-coloring in \cite{BeersMulas2025}. For an $h$-uniform hypergraph, the canonical notion of $k$-coloring is equivalent to the $(h-1)$-bounded $k$-coloring defined here.

\begin{definition}[The $(h,a,k,\varepsilon)$-Balanced Coloring Problem]\label{def:bch}
    Fix integers $h,k \ge 2$, $1 \le a < h$, and real $0<\varepsilon \le 1/k$. $(h,a,k,\varepsilon)$-Balanced Coloring asks to distinguish between the following two cases for an $h$-uniform hypergraph $H=(V,E)$:
    \begin{itemize}
        \item (Yes) $H$ is $(a,k)$-colorable. %
        \item (No) Any independent set in $H$ is of size $<\varepsilon \cdot |V|$.
    \end{itemize}
\end{definition}

%% file: sec/reduction.tex
\section{A Reduction from Hypergraph Coloring to Unit-Length UMPS}
\label{sec:reduction}

\begin{theorem}\label{thm:main_reduction}
    Fix integers $h,k \ge 2$, $1 \le a < h$, and real $0<\varepsilon \le 1/k$. There is a polynomial-time reduction from an $(h,a,k,\varepsilon)$-Balanced Coloring instance $H$ to a unit-length UMPS instance $\Gamma$, such that
    \begin{itemize}
        \item In the Yes case, $\operatorname{OPT}(\Gamma) \le (k+2c)a$.
        \item In the No case, $\operatorname{OPT}(\Gamma) \ge ka+c(a+h).$
    \end{itemize}
\end{theorem}

It is useful to think of having $k$ virtual copies $H^{(1)},\ldots,H^{(k)}$ of the input hypergraph $H=(V,E)$. In the yes case, there is a balanced coloring of $H$, and thus we can make a cyclic shift on the coloring of $H^{(1)},\ldots,H^{(k)}$, such that the $k$ copies of each vertex are assigned $k$ different colors. In the no case, even if we only consider a partial coloring, where for each vertex, one of its $k$ copies is assigned a fixed color (say color $k$) while the others are left uncolored, in some hypergraph copy the set of colored vertices is too large to stay independent, yielding a monochromatic hyperedge.

The physical machines are indexed by a layer $t \in \{0,\ldots, c\}$.

\subsection{Construction}

\paragraph{Vertex Machines.}

For every vertex $v \in V$ and every layer $t \in \{0,\ldots,c\}$, create one machine $M_{v,t}$. For every virtual copy $i \in [k]$, put on $M_{v,t}$ a chain of $a$ unit jobs
$$J_{v,t}^{i,1} \prec J_{v,t}^{i,2}\prec \ldots \prec J_{v,t}^{i,a}.$$

We write
$$S_{v,i,t}:=J_{v,t}^{i,1} \quad \text{and} \quad T_{v,i,t}:=J_{v,t}^{i,a}$$
for the head and tail of this chain. Thus each vertex machine contains exactly $k$ chains of length $a$. 

\begin{figure}[H]
\centering
\begin{tikzpicture}[x=1cm,y=1cm,>=Latex,font=\small]
  \def\machineX{3.45}
  \def\slotW{2.62}
  \def\slotH{0.96}
  \def\topY{6.75}
  \def\chainY{5.17}
  \def\copyY{4.55}
  \def\bottomY{1.85}
  \def\bottomCopyY{0.98}

  \node[anchor=west,font=\bfseries] at (0.15,8.02)
    {(a) Schematic of a vertex machine: four slots and four virtual-copy chains};

  \node[anchor=east] at (3.10,{\topY+0.48}) {$M_{v,t}$};

  \foreach \r/\col in {1/slotred,2/slotblue,3/slotgreen,4/slotyellow}{
    \pgfmathsetmacro{\xx}{\machineX+(\r-1)*\slotW}
    \path[fill=\col!13] (\xx,\topY) rectangle ({\xx+\slotW},{\topY+\slotH});
    \node[font=\tiny,anchor=north east]
      at ({\xx+\slotW-0.06},{\topY+\slotH-0.04}) {$\r$};
  }
  \draw[line width=0.95pt]
    (\machineX,\topY) rectangle ({\machineX+4*\slotW},{\topY+\slotH});
  \foreach \r in {1,2,3}{
    \pgfmathsetmacro{\divider}{\machineX+\r*\slotW}
    \draw[line width=0.75pt] (\divider,\topY)--(\divider,{\topY+\slotH});
  }
  \draw[<->,thin]
    (\machineX,{\topY-0.22}) -- ({\machineX+\slotW},{\topY-0.22})
    node[midway,below=1pt,font=\scriptsize] {one length-$a$ slot};

  \node[anchor=east] at (3.10,{\chainY+0.38}) {chains of jobs};
  \node[anchor=east] at (3.10,\copyY) {virtual copies of $v$};

  \foreach \i in {1,2,3,4}{
    \pgfmathsetmacro{\cx}{\machineX+(\i-0.5)*\slotW}
    \node[draw,rounded corners=1.5pt,minimum width=2.15cm,minimum height=0.76cm,
          fill=white,inner sep=2pt] at (\cx,{\chainY+0.38}) {$J^{\i,*}_{v,t}$};
    \node[circle,draw,minimum size=0.76cm,inner sep=1pt] at (\cx,\copyY)
      {$v^{(\i)}$};
  }

  \draw[gray!55,line width=0.65pt] (0.15,3.62)--(14.05,3.62);

  \node[anchor=west,font=\bfseries] at (0.15,3.27)
    {(b) One ordering of the job chains on the same vertex machine $M_{v,t}$};

  \node[anchor=east] at (3.10,{\bottomY+0.48}) {$M_{v,t}$};

  \foreach \slot/\chain/\col in {1/2/slotred,2/3/slotblue,3/4/slotgreen,4/1/slotyellow}{
    \pgfmathsetmacro{\xx}{\machineX+(\slot-1)*\slotW}
    \pgfmathsetmacro{\cx}{\xx+0.5*\slotW}
    \path[fill=\col!13] (\xx,\bottomY) rectangle ({\xx+\slotW},{\bottomY+\slotH});
    \node[draw=black,rounded corners=1.5pt,minimum width=1.92cm,minimum height=0.64cm,
          fill=white,inner sep=1.5pt] at (\cx,{\bottomY+0.5*\slotH})
      {$J^{\chain,*}_{v,t}$};
    \node[font=\tiny,anchor=north east]
      at ({\xx+\slotW-0.06},{\bottomY+\slotH-0.04}) {$\slot$};
  }
  \draw[line width=0.95pt]
    (\machineX,\bottomY) rectangle ({\machineX+4*\slotW},{\bottomY+\slotH});
  \foreach \r in {1,2,3}{
    \pgfmathsetmacro{\divider}{\machineX+\r*\slotW}
    \draw[line width=0.75pt] (\divider,\bottomY)--(\divider,{\bottomY+\slotH});
  }

  \node[anchor=east,align=right] at (3.10,\bottomCopyY)
    {encoded colors\\[-1pt](copies in index order)};

  \foreach \i/\col/\r in {1/slotyellow/4,2/slotred/1,3/slotblue/2,4/slotgreen/3}{
    \pgfmathsetmacro{\cx}{\machineX+(\i-0.5)*\slotW}
    \node[circle,draw=\col!78!black,fill=\col!12,line width=0.9pt,
          minimum size=0.82cm,inner sep=1pt] at (\cx,\bottomCopyY) {$v^{(\i)}$};
    \node[font=\scriptsize,text=\col!65!black] at (\cx,{\bottomCopyY-0.62}) {color $\r$};
  }
\end{tikzpicture}
\caption{The vertex-machine gadget for $k=4$. The four length-$a$ slots represent colors $1,2,3,4$ (red, blue, green, and yellow). A chain has no intrinsic color; its position on $M_{v,t}$ determines the color of the corresponding virtual copy. In panel (b), the order $J^{2,*}_{v,t},J^{3,*}_{v,t},J^{4,*}_{v,t},J^{1,*}_{v,t}$ encodes colors $4,1,2,3$ for $v^{(1)},v^{(2)},v^{(3)},v^{(4)}$, respectively.}
\label{fig:vertex-machine}
\end{figure}

The intended encoding is as follows:
\begin{quote}
    If chain $i$ occupies the $r$-th length-$a$ slot on $M_{v,t}$, then the virtual copy $v^{(i)}$ is interpreted as having color $r$ at layer $t$.
\end{quote}

\paragraph{Hyperedge Machines.}

For every hyperedge $e \in E$ and every $i \in [k],t \in [c]$, create one machine $M_{e,i,t}$, which encodes the balancing constraint of the virtual copy $e^{(i)}$. For every $v \in e$, we put one unit job $Y_{e,i,t}^v$ on this machine, and add precedence constraints
$$
T_{v,i,t-1} \prec Y_{e,i,t}^v\prec S_{v,i,t}.
$$
Since $|e|=h$, each machine $M_{e,i,t}$ contains exactly $h$ jobs.

It's easy to see the whole reduction runs in polynomial time.

\subsection{Completeness}
Assume the Yes case and let $\chi:V \to [k]$ be a balanced coloring. For each virtual copy $i \in [k]$, we define the cyclic shift
$$
\chi_i(v):=(\chi(v)+i)\bmod k+1.
$$
For each fixed $v \in V$, the values $\chi_1(v),\ldots, \chi_k(v)$ form a permutation of $[k]$. For each $i \in [k]$, the coloring $\chi_i$ is a relabeling of $\chi$, so we still have
$$|e \cap \chi_i^{-1}(r)|\le a$$
for every $e \in E$ and $r \in [k]$.

\begin{lemma}[Completeness]
    In the Yes case, the constructed UMPS instance has a feasible schedule of makespan at most
    $$C_Y=(k+2c)a.$$
\end{lemma}
\begin{proof}
    Fix $v \in V$, a layer $ t \in \{0,\ldots,c\}$, and a virtual copy $i \in [k]$. Put the entire chain $J_{v,t}^{i,1},\ldots,J_{v,t}^{i,a}$ in the interval
    $$[(\chi_i(v)-1+2t)a,(\chi_i(v)+2t)a).$$
    Since the $k$ values $\chi_1(v),\ldots,\chi_k(v)$ form a permutation, the $k$ chains exactly fill $[2ta,(k+2t)a)$ on $M_{v,t}$. Furthermore, for every vertex, each of the $c+1$ layers uses the same permutation of chains.

    Fix $t \in \{1,\ldots,c\}$ and $r \in [k]$, define the window
    $$W_{t,r}=[(r+2t-2)a, (r+2t-1)a).$$
    The windows $W_{t,1},\ldots,W_{t,k}$ are consecutive and disjoint on each edge machine.
    
    Fix $M_{e,i,t}$. Schedule all jobs $Y_{e,i,t}^v$ with $\chi_i(v)=r$ in window $W_{t,r}$. There are at most $a$ of them, so they fit. For any $v \in e$ with $\chi_i(v)=r$, the tail $T_{v,i,t-1}$ completes exactly at the left endpoint of $W_{t,r}$, while the head $S_{v,i,t}$ starts exactly at its right endpoint. Therefore, all precedence constraints are satisfied. The last vertex job finishes at $(k+2c)a$ and every edge job finishes earlier.
\end{proof}

\begin{figure}[H]
\centering
\begin{tikzpicture}[x=1cm,y=1cm,>=Latex,font=\scriptsize]
  \def\slotW{1.78}
  \def\machineH{0.78}
  \def\edgeH{0.78}
  \def\chainW{1.46}
  \def\xzero{2.55}
  \def\labelX{2.22}
  \pgfmathsetmacro{\edgeX}{\xzero+\slotW}
  \pgfmathsetmacro{\bottomX}{\xzero+2*\slotW}

  \def\topone{10.90}
  \def\toptwo{10.00}
  \def\topthree{9.10}
  \def\topfour{8.20}
  \def\topfive{7.30}
  \def\edgeY{6.02}
  \def\bottomone{4.18}
  \def\bottomtwo{3.28}
  \def\bottomthree{2.38}
  \def\bottomfour{1.48}
  \def\bottomfive{0.58}
  \pgfmathsetmacro{\figureCenter}{\edgeX+2*\slotW}

  \node[anchor=center,font=\bfseries\small] at (\figureCenter,12.62)
    {Completeness schedule around one hyperedge machine};
  \draw[->,line width=0.65pt]
    (\xzero,12.18) -- ({\xzero+6*\slotW},12.18)
    node[right,font=\scriptsize] {time};

  \foreach \v/\r/\yy in {1/2/\topone,2/1/\toptwo,3/1/\topthree,4/3/\topfour,5/4/\topfive}{
    \node[anchor=east] at (\labelX,{\yy+0.5*\machineH}) {$M_{v_\v,t-1}$};
    \foreach \s/\col in {1/slotred,2/slotblue,3/slotgreen,4/slotyellow}{
      \pgfmathsetmacro{\sx}{\xzero+(\s-1)*\slotW}
      \path[fill=\col!13] (\sx,\yy) rectangle ({\sx+\slotW},{\yy+\machineH});
    }
    \draw[line width=0.78pt] (\xzero,\yy) rectangle ({\xzero+4*\slotW},{\yy+\machineH});
    \foreach \s in {1,2,3}{
      \pgfmathsetmacro{\sx}{\xzero+\s*\slotW}
      \draw[line width=0.55pt] (\sx,\yy)--(\sx,{\yy+\machineH});
    }
  }

  \foreach \v/\r/\yy in {1/2/\bottomone,2/1/\bottomtwo,3/1/\bottomthree,4/3/\bottomfour,5/4/\bottomfive}{
    \node[anchor=east] at (\labelX,{\yy+0.5*\machineH}) {$M_{v_\v,t}$};
    \foreach \s/\col in {1/slotred,2/slotblue,3/slotgreen,4/slotyellow}{
      \pgfmathsetmacro{\sx}{\bottomX+(\s-1)*\slotW}
      \path[fill=\col!13] (\sx,\yy) rectangle ({\sx+\slotW},{\yy+\machineH});
    }
    \draw[line width=0.78pt] (\bottomX,\yy) rectangle ({\bottomX+4*\slotW},{\yy+\machineH});
    \foreach \s in {1,2,3}{
      \pgfmathsetmacro{\sx}{\bottomX+\s*\slotW}
      \draw[line width=0.55pt] (\sx,\yy)--(\sx,{\yy+\machineH});
    }
  }

  \node[anchor=east] at (\labelX,{\edgeY+0.5*\edgeH}) {$M_{e,i,t}$};
  \foreach \r/\col in {1/slotred,2/slotblue,3/slotgreen,4/slotyellow}{
    \pgfmathsetmacro{\wx}{\edgeX+(\r-1)*\slotW}
    \path[fill=\col!13] (\wx,\edgeY) rectangle ({\wx+\slotW},{\edgeY+\edgeH});
    \node[anchor=north,font=\scriptsize] at ({\wx+0.5*\slotW},{\edgeY-0.12}) {$W_{t,\r}$};
  }
  \draw[line width=0.9pt] (\edgeX,\edgeY) rectangle ({\edgeX+4*\slotW},{\edgeY+\edgeH});
  \foreach \r in {1,2,3}{
    \pgfmathsetmacro{\wx}{\edgeX+\r*\slotW}
    \draw[line width=0.62pt] (\wx,\edgeY)--(\wx,{\edgeY+\edgeH});
  }
  \draw[<->,line width=0.55pt]
    (\edgeX,{\edgeY-0.63}) -- ({\edgeX+\slotW},{\edgeY-0.63})
    node[midway,below=1pt,font=\scriptsize] {length $a$};

  \pgfmathsetmacro{\jobvTwoX}{\edgeX+0.28*\slotW}
  \pgfmathsetmacro{\jobvThreeX}{\edgeX+0.68*\slotW}
  \pgfmathsetmacro{\jobvOneX}{\edgeX+1.28*\slotW}
  \pgfmathsetmacro{\jobvFourX}{\edgeX+2.28*\slotW}
  \pgfmathsetmacro{\jobvFiveX}{\edgeX+3.28*\slotW}
  \pgfmathsetmacro{\edgeMidY}{\edgeY+0.5*\edgeH}
  \pgfmathsetmacro{\edgeTopJobY}{\edgeY+0.64}
  \pgfmathsetmacro{\edgeBottomJobY}{\edgeY+0.14}

  \foreach \v/\r/\col/\topY/\botY/\jobX in {
    1/2/slotblue/\topone/\bottomone/\jobvOneX,
    2/1/slotred/\toptwo/\bottomtwo/\jobvTwoX,
    3/1/slotred/\topthree/\bottomthree/\jobvThreeX,
    4/3/slotgreen/\topfour/\bottomfour/\jobvFourX,
    5/4/slotyellow/\topfive/\bottomfive/\jobvFiveX}{
      \pgfmathsetmacro{\topTailX}{\xzero+(\r-0.5)*\slotW+0.5*\chainW}
      \pgfmathsetmacro{\bottomHeadX}{\bottomX+(\r-0.5)*\slotW-0.5*\chainW}
      \pgfmathsetmacro{\topMidY}{\topY+0.5*\machineH}
      \pgfmathsetmacro{\bottomMidY}{\botY+0.5*\machineH}
      \draw[-{Latex[length=1.45mm,width=1.15mm]},draw=\col,
            line width=1.00pt,rounded corners=2.6pt]
        (\topTailX,\topMidY) -- (\jobX,\topMidY) -- (\jobX,\edgeTopJobY);
      \draw[-{Latex[length=1.45mm,width=1.15mm]},draw=\col,
            line width=1.00pt,rounded corners=2.6pt]
        (\jobX,\edgeBottomJobY) -- (\jobX,\bottomMidY) -- (\bottomHeadX,\bottomMidY);
  }

  \foreach \v/\r/\yy/\col in {1/2/\topone/slotblue,2/1/\toptwo/slotred,3/1/\topthree/slotred,4/3/\topfour/slotgreen,5/4/\topfive/slotyellow}{
    \pgfmathsetmacro{\cx}{\xzero+(\r-0.5)*\slotW}
    \node[draw=\col!90!black,fill=white,line width=0.60pt,rounded corners=1.2pt,
          minimum width=\chainW cm,minimum height=0.46cm,
          inner xsep=1pt,inner ysep=0.7pt,
          font=\fontsize{6.2}{7.2}\selectfont]
      at (\cx,{\yy+0.5*\machineH}) {$J^{i,*}_{v_\v,t-1}$};
  }
  \foreach \v/\r/\yy/\col in {1/2/\bottomone/slotblue,2/1/\bottomtwo/slotred,3/1/\bottomthree/slotred,4/3/\bottomfour/slotgreen,5/4/\bottomfive/slotyellow}{
    \pgfmathsetmacro{\cx}{\bottomX+(\r-0.5)*\slotW}
    \node[draw=\col!90!black,fill=white,line width=0.60pt,rounded corners=1.2pt,
          minimum width=\chainW cm,minimum height=0.46cm,
          inner xsep=1pt,inner ysep=0.7pt,
          font=\fontsize{6.2}{7.2}\selectfont]
      at (\cx,{\yy+0.5*\machineH}) {$J^{i,*}_{v_\v,t}$};
  }

  \foreach \x/\v/\col in {
    \jobvTwoX/2/slotred,
    \jobvThreeX/3/slotred,
    \jobvOneX/1/slotblue,
    \jobvFourX/4/slotgreen,
    \jobvFiveX/5/slotyellow}{
      \node[draw=\col!90!black,fill=white,rounded corners=1.1pt,
            minimum width=0.66cm,minimum height=0.50cm,inner xsep=0.5pt,inner ysep=0.8pt,
            font=\scriptsize]
        at (\x,\edgeMidY) {$Y^{v_\v}$};
  }

  \foreach \r/\col in {1/slotred,2/slotblue,3/slotgreen,4/slotyellow}{
    \pgfmathsetmacro{\tx}{\xzero+(\r-0.5)*\slotW}
    \node[text=\col!70!black,font=\tiny] at (\tx,{\topone+\machineH+0.18}) {color \r};
    \pgfmathsetmacro{\bx}{\bottomX+(\r-0.5)*\slotW}
    \node[text=\col!70!black,font=\tiny] at (\bx,{\bottomone+\machineH+0.18}) {color \r};
  }
\end{tikzpicture}
\caption{An illustrative completeness schedule for one hyperedge machine, with $k=4$, $e=\{v_1,\ldots,v_5\}$, and $a\ge 2$. Here $\chi_i(v_2)=\chi_i(v_3)=1$, $\chi_i(v_1)=2$, $\chi_i(v_4)=3$, and $\chi_i(v_5)=4$. The edge jobs are grouped by color: each $Y^{v_j}_{e,i,t}$ is placed in $W_{t,\chi_i(v_j)}$. The upper and lower copies of $J^{i,*}_{v_j,\cdot}$ occupy the same color slot on the two adjacent vertex layers. Each colored path depicts the two precedence constraints $T_{v_j,i,t-1}\prec Y^{v_j}_{e,i,t}\prec S_{v_j,i,t}$.}
\label{fig:edge-machine-completeness}
\end{figure}
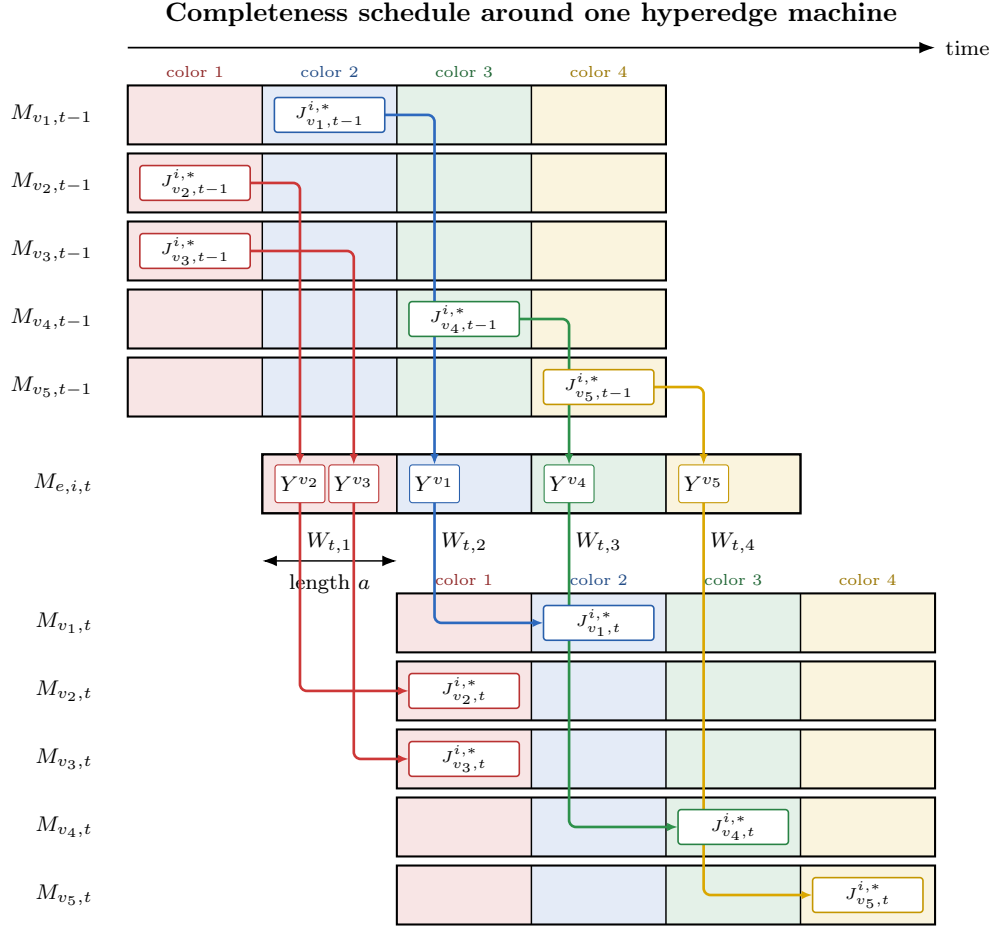

\begin{remark}
    The chain-slot interpretation is only used to construct the schedule in the Yes case. In the No case we do \textit{not} assume that the schedule keeps each chain contiguous, nor do we try to decode a complete $k$-coloring from its chain order.
\end{remark}

\subsection{Soundness}

Assume the No case. Recall that we have $\varepsilon \le \frac{1}{kc}$. For $t \in \{1,\ldots,c+1\}$, let
$$B_t =ka+(t-1)(a+h),$$
and for $t \in \{1,\ldots,c\}$, define the canonical interval at stage $t$ to be
$$I_t=[B_t,B_t+h).$$
$I_t$ can be viewed as the no-case analogue of the $W_{t,k}$ in the yes case. The difference here is that we assume each window is of length $h$, where in the yes case we only need length $a$. The proof idea is that, even though the vertex chains may not follow the way we schedule them in the yes case, and their jobs may be delayed or interleaved arbitrarily; any delay will only squeeze the room for later layers, implying that for some $t \in \{1,\ldots,c\}$ the actual usable window for $M_{e,i,t}$ is contained in the canonical window $I_t$.

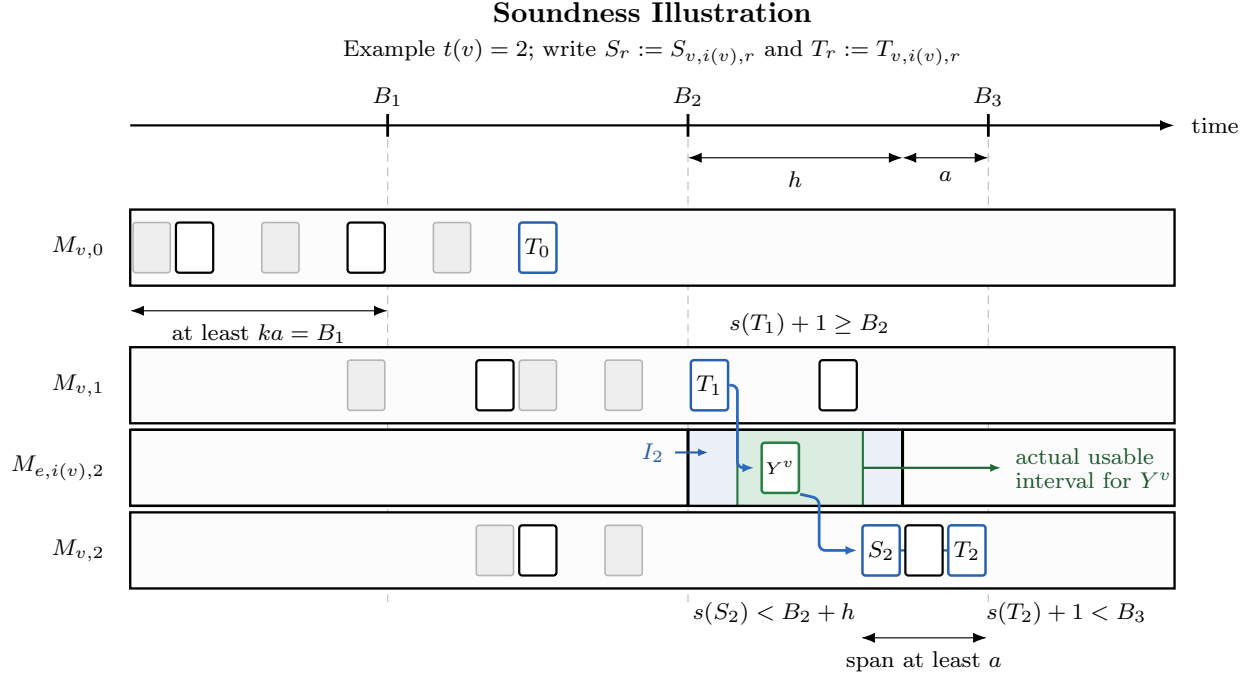
\begin{figure}[H]
\centering
\resizebox{\textwidth}{!}{%
\begin{tikzpicture}[x=1cm,y=1cm,>=Latex,font=\scriptsize]
  \def\xL{1.45}
  \def\xR{14.35}
  \def\u{0.53}
  \pgfmathsetmacro{\BOne}{\xL+6*\u}
  \pgfmathsetmacro{\BTwo}{\BOne+7*\u}
  \pgfmathsetmacro{\BTwoH}{\BTwo+5*\u}
  \pgfmathsetmacro{\BThree}{\BTwo+7*\u}
  \pgfmathsetmacro{\xMid}{0.5*(\xL+\xR)}

  \def\axisY{8.15}
  \def\yZero{6.64}
  \def\yOne{4.94}
  \def\yEdge{3.92}
  \def\yTwo{2.90}
  \def\trackH{0.94}
  \def\jobW{0.46}
  \def\jobH{0.62}

  \tikzset{
    machine/.style={draw=black,fill=gray!2,line width=0.82pt},
    selectedjob/.style={draw=slotblue!92!black,fill=white,line width=0.88pt,
                        rounded corners=1.1pt,minimum width=\jobW cm,
                        minimum height=\jobH cm,text width=0.40cm,
                        align=center,inner sep=0pt},
    tailjob/.style={draw=black,fill=white,line width=0.78pt,
                    rounded corners=1.1pt,minimum width=\jobW cm,
                    minimum height=\jobH cm,text width=0.40cm,
                    align=center,inner sep=0pt},
    otherjob/.style={draw=gray!58,fill=gray!13,line width=0.56pt,
                     rounded corners=1.0pt,minimum width=\jobW cm,
                     minimum height=\jobH cm,text width=0.40cm,
                     align=center,inner sep=0pt},
    yjob/.style={draw=slotgreen!88!black,fill=white,line width=0.88pt,
                 rounded corners=1.1pt,minimum width=\jobW cm,
                 minimum height=\jobH cm,text width=0.40cm,
                 align=center,inner sep=0pt,font=\tiny}
  }

  \node[anchor=center,font=\bfseries\small] at (\xMid,9.55)
    {Soundness Illustration};
  \node[anchor=center] at (\xMid,9.06)
    {Example $t(v)=2$; write $S_r:=S_{v,i(v),r}$ and $T_r:=T_{v,i(v),r}$};

  \draw[-{Latex[length=1.8mm,width=1.35mm]},line width=0.78pt]
    (\xL,\axisY)--(\xR,\axisY) node[right=2pt] {time};
  \foreach \xx/\lab in {\BOne/B_1,\BTwo/B_2,\BThree/B_3}{
    \draw[line width=0.98pt] (\xx,{\axisY-0.14})--(\xx,{\axisY+0.14});
    \node[anchor=south] at (\xx,{\axisY+0.10}) {$\lab$};
    \draw[densely dashed,gray!45,line width=0.46pt]
      (\xx,{\axisY-0.17})--(\xx,{\yTwo-0.64});
  }

  \draw[<->,thin] (\BTwo,{\axisY-0.40})--(\BTwoH,{\axisY-0.40})
    node[midway,below=1pt] {$h$};
  \draw[<->,thin] (\BTwoH,{\axisY-0.40})--(\BThree,{\axisY-0.40})
    node[midway,below=1pt] {$a$};

  \foreach \yy in {\yZero,\yOne,\yEdge,\yTwo}{
    \draw[machine] (\xL,{\yy-0.5*\trackH}) rectangle (\xR,{\yy+0.5*\trackH});
  }
  \node[anchor=east] at ({\xL-0.18},\yZero) {$M_{v,0}$};
  \node[anchor=east] at ({\xL-0.18},\yOne) {$M_{v,1}$};
  \node[anchor=east] at ({\xL-0.18},\yEdge) {$M_{e,i(v),2}$};
  \node[anchor=east] at ({\xL-0.18},\yTwo) {$M_{v,2}$};

  \node[otherjob] at ({\xL+0.5*\u},\yZero) {};
  \node[tailjob]  at ({\xL+1.5*\u},\yZero) {};
  \node[otherjob] at ({\xL+3.5*\u},\yZero) {};
  \node[tailjob]  at ({\xL+5.5*\u},\yZero) {};
  \node[otherjob] at ({\xL+7.5*\u},\yZero) {};
  \node[selectedjob] (tailzero) at ({\xL+9.5*\u},\yZero) {$T_0$};

  \draw[<->,thin] (\xL,{\yZero-0.78})--(\BOne,{\yZero-0.78})
    node[midway,below=1pt] {at least $ka=B_1$};

  \node[otherjob] at ({\xL+5.5*\u},\yOne) {};
  \node[tailjob]  at ({\xL+8.5*\u},\yOne) {};
  \node[otherjob] at ({\xL+9.5*\u},\yOne) {};
  \node[otherjob] at ({\xL+11.5*\u},\yOne) {};
  \node[selectedjob] (tailone) at ({\xL+13.5*\u},\yOne) {$T_1$};
  \node[tailjob] at ({\xL+16.5*\u},\yOne) {};

  \node[anchor=south west] at ({\BTwo+0.38},{\yOne+0.48})
    {$s(T_1)+1\ge B_2$};

  \path[fill=slotblue!9]
    (\BTwo,{\yEdge-0.5*\trackH}) rectangle (\BTwoH,{\yEdge+0.5*\trackH});
  \draw[black,line width=1.08pt]
    (\BTwo,{\yEdge-0.5*\trackH})--(\BTwo,{\yEdge+0.5*\trackH});
  \draw[black,line width=1.08pt]
    (\BTwoH,{\yEdge-0.5*\trackH})--(\BTwoH,{\yEdge+0.5*\trackH});

  \def\actualL{8.95}
  \def\actualR{10.50}
  \path[fill=slotgreen!17]
    (\actualL,{\yEdge-0.5*\trackH}) rectangle (\actualR,{\yEdge+0.5*\trackH});
  \draw[draw=slotgreen!78!black,line width=0.72pt]
    (\actualL,{\yEdge-0.5*\trackH})--(\actualL,{\yEdge+0.5*\trackH});
  \draw[draw=slotgreen!78!black,line width=0.72pt]
    (\actualR,{\yEdge-0.5*\trackH})--(\actualR,{\yEdge+0.5*\trackH});
  \node[yjob] (edgejob) at (9.48,\yEdge) {$Y^v$};

  \node[text=slotblue!88!black,anchor=east] (ilabel)
    at ({\BTwo-0.16},{\yEdge+0.19}) {$I_2$};
  \draw[-{Latex[length=1.25mm,width=1.0mm]},draw=slotblue!88!black,
        line width=0.72pt]
    (ilabel.east) -- ({\BTwo+0.25},{\yEdge+0.19});

  \node[text=slotgreen!60!black,anchor=west,align=left] (actualLabel)
    at ({\BThree+0.20},\yEdge)
    {actual usable\\[-1pt]interval for $Y^v$};
  \draw[-{Latex[length=1.35mm,width=1.05mm]},draw=slotgreen!72!black,
        line width=0.78pt,shorten >=1.2pt]
    (\actualR,\yEdge) -- (actualLabel.west);

  \draw[black,line width=0.82pt]
    (\xL,{\yEdge-0.5*\trackH})--(\xR,{\yEdge-0.5*\trackH});
  \draw[black,line width=0.82pt]
    (\xL,{\yEdge+0.5*\trackH})--(\xR,{\yEdge+0.5*\trackH});

  \node[otherjob] at ({\xL+8.5*\u},\yTwo) {};
  \node[tailjob]  at ({\xL+9.5*\u},\yTwo) {};
  \node[otherjob] at ({\xL+11.5*\u},\yTwo) {};
  \draw[draw=slotblue!90!black,line width=0.78pt]
    ({\xL+17.5*\u},\yTwo)--({\xL+19.5*\u},\yTwo);
  \node[selectedjob] (headtwo) at ({\xL+17.5*\u},\yTwo) {$S_2$};
  \node[tailjob] at ({\xL+18.5*\u},\yTwo) {};
  \node[selectedjob] (tailtwo) at ({\xL+19.5*\u},\yTwo) {$T_2$};

  \draw[-{Latex[length=1.45mm,width=1.15mm]},draw=slotblue,
        line width=0.98pt,rounded corners=2.2pt,shorten >=2.2pt]
    (tailone.east) -- (\actualL,\yOne) -- (\actualL,\yEdge) -- (edgejob.west);
  \draw[-{Latex[length=1.45mm,width=1.15mm]},draw=slotblue,
        line width=0.98pt,rounded corners=2.2pt,shorten >=1.4pt]
    (edgejob.south east) -- (9.96,{\yEdge-0.38}) -- (9.96,\yTwo) -- (headtwo.west);

  \node[anchor=north east,align=right] at ({\xL+17.15*\u},{\yTwo-0.50})
    {$s(S_2)<B_2+h$};
  \node[anchor=north west] at ({\xL+19.70*\u},{\yTwo-0.50})
    {$s(T_2)+1<B_3$};

  \draw[<->,thin]
    (headtwo.west|-{0,\yTwo-1.08})--(tailtwo.east|-{0,\yTwo-1.08})
    node[midway,below=1pt] {span at least $a$};
\end{tikzpicture}%
}
\caption{Soundness argument for one vertex, illustrated when $t(v)=2$. The selected tails satisfy $s(T_1)+1\ge B_2$ and $s(T_2)+1<B_3$, while the length-$a$ selected chain forces $s(S_2)<B_2+h$. The two precedence constraints therefore squeeze $Y^v_{e,i(v),2}$ into the highlighted stage-$2$ window.}
\label{fig:soundness-illustration}
\end{figure}

\begin{lemma}[Soundness]
    In the No case, any feasible schedule has makespan at least 
    $$C_N=ka+c(a+h).$$
\end{lemma}
\begin{proof}
    Consider an arbitrary feasible schedule and $C_{\max}$ denote its makespan. Suppose for contradiction that $C_{\max} < ka+c(a+h)$. For each vertex $v$, choose $i(v) \in [k]$ maximizing the completion time of $T_{v,i,0}$, the tail of the $i$-th length-$a$ chain:
    $$i(v):=\arg\max_{i \in [k]}\{s(T_{v,i,0})+1\}.$$
    The $ka$ unit jobs on $M_{v,0}$ all finish by the last chain tail, so
    $$s(T_{v,i(v),0})+1 \ge ka = B_1.$$
    Now let $t(v)$ be the first stage that this specific chain tail ends before $B_{t+1}$:
    $$t(v):=\min\{t \in [c]: s(T_{v,i(v),t})+1<B_{t+1}\}.$$
    This stage exists because all jobs are completed by $C_{\max}<B_{c+1}$.
    By the minimality of $t(v)$, together with the layer-0 inequality above when $t(v)=1$, we know the completion time of the chain tail in the previous layer is at least $B_{t(v)}$:
    $$s(T_{v,i(v),t(v)-1})+1 \ge B_{t(v)}.$$
    The starting time of the head of this chain is at most
    $$s(S_{v,i(v),t(v)}) \le s(T_{v,i(v),t(v)})-(a-1)<B_{t(v)+1}-a = B_{t(v)}+h.$$
    For every hyperedge $e \ni v$, the precedence constraints
    $$T_{v,i(v),t(v)-1} \prec Y_{e,i(v),t(v)}^v \prec S_{v,i(v),t(v)}$$
    therefore imply $s(Y_{e,i(v),t(v)}^v)\in[B_{t(v)},B_{t(v)}+h-1)$.

    Now, each vertex has selected one of the $kc$ pairs $(i(v),t(v))$. By the pigeonhole principle, some pair $(i,t)$ is selected by at least 
    $$\frac{|V|}{kc} \ge \varepsilon |V|$$
    vertices. This set is not independent in the No case, so it contains a hyperedge $e \in E$. However, this means $h=|e|$ jobs are assigned in the interval $[B_{t},B_t+h-1)$ on the same machine $M_{e,i,t}$, a contradiction.
\end{proof}

\subsection{Unit-Length UMPS Hardness}
For fixed $h,a,k$, the gap in \Cref{thm:main_reduction} satisfies 
$$\lim_{c \to \infty}\left(\frac{ka+c(a+h)}{(k+2c)a}\right) = \frac{a+h}{2a}.$$

We now prove our main theorems (\Cref{thm:np} and \Cref{thm:quasi}) by plugging suitable hypergraph coloring hardness into \Cref{thm:main_reduction}.

In \cite{Guruswami2018}, the authors proved that in the yes case in any hyperedge, any color appears at least $Q-1$ times. Since the hypergraph is $Qk$-uniform, this implies the most frequent color appears at most \[
Qk-(k-1)(Q-1)=Q+k-1
\] times in every hyperedge.

\begin{lemma}[Adapted from {\protect\cite[Theorem 1.1]{Guruswami2018}}]
\label{thm:hardness-hypergraph-coloring}
Assume $\classp\neq\classnp$. For any constants $\varepsilon>0$ and $Q,k\geq 2$, no algorithm can solve $(Qk, Q+k-1,k,\varepsilon)$-Balanced Coloring in polynomial time.
\end{lemma}

\thmnp*
\begin{proof}[Proof of \Cref{thm:np}]
    Choose $k$ such that $(k+1)/2>\gamma$. Since for $a=Q+k-1$ and $h=Qk$,
    $$\lim_{Q \to \infty}\lim_{c \to \infty} \frac{ka+c(a+h)}{(k+2c)a}=\frac{k+1}{2},$$
    we can choose $Q,c$ large enough so that the ratio is larger than $\gamma$. Setting $\varepsilon = \frac{1}{kc}$ and plugging the parameters into \Cref{thm:hardness-hypergraph-coloring} and \Cref{thm:main_reduction} finish the proof.
\end{proof}

For our polylogarithmic inapproximability, we start with the following result of Guruswami, Harsha, H\r{a}stad, Srinivasan, and Varma \cite[Theorem~1.2]{GuruswamiHarshaHastadSrinivasanVarma2017}%
, and use \Cref{thm:coloring-amplification} to amplify the parameters of our interest. 

\begin{restatable}[Hardness of coloring a $4$-colorable $4$-uniform hypergraph \cite{GuruswamiHarshaHastadSrinivasanVarma2017}]{lemma}{ghhfour}
\label{lem:hardness-4-coloring}
There is a quasi-polynomial time reduction from a $3\mathrm{SAT}$ instance $\Gamma$ of size $n$ to a $4$-uniform hypergraph $G=(V,E)$ with $n'$ vertices, where $n\le n'\le n^{\mathrm{polylog}(n)}$, such that:
\begin{itemize}
    \item (Yes) If $\Gamma$ is satisfiable, then $G$ is $4$-colorable.
    \item (No) If $\Gamma$ is unsatisfiable, then every independent set in $G$ has size less than
    \[\frac{n'}{2^{2^{\Omega(\sqrt{\log\log n'})}}}.
    \]
\end{itemize}
\end{restatable}

\begin{restatable}[Balanced Coloring Amplification]{theorem}{coloringamplification}
\label{thm:coloring-amplification}
Let $G=(V,E)$ be an $h$-uniform hypergraph with $n=|V|$ vertices and $m=|E|$ hyperedges. For every integer $t\ge 1$ and every real $0<\varepsilon<1/t$, there is a  deterministic polynomial-time reduction that constructs an $h^t$-uniform hypergraph $G^{(t)}=(V^{(t)},E^{(t)})$ such that
\begin{itemize}
    \item $|V^{(t)}|=n^t$ and $|E^{(t)}|=m^{1+h+\cdots+h^{t-1}}$;
    \item if $G$ is $(a,k)$-colorable, then $G^{(t)}$ is $(a^t,k^t)$-colorable;
    \item if every independent set in $G$ has size less than $\varepsilon |V|$, then every independent set in $G^{(t)}$ has size less than $t\varepsilon |V^{(t)}|$.
\end{itemize}
\end{restatable}

The proof of \Cref{thm:coloring-amplification} is deferred to \Cref{sec:hardness-coloring}. Now we prove \Cref{thm:quasi} by instantiating \Cref{thm:coloring-amplification} with \Cref{lem:hardness-4-coloring}.

\thmquasi*
\begin{proof}[Proof of \Cref{thm:quasi}]
    Let $\Gamma$ be a $3\mathrm{SAT}$ instance of size $n$. By \Cref{lem:hardness-4-coloring}, we obtain a $4$-uniform hypergraph $G=(V,E)$ with $n'=|V|$ vertices, where $n\le n'\le n^{\mathrm{polylog}(n)}$. There is a constant $\beta>0$ such that, in the No case, every independent set in $G$ has size less than $\alpha n'$, where
    \[
        \alpha:=2^{-2^{\beta\sqrt{\log\log n'}}}.
    \]

    Fix a sufficiently small constant $\eta>0$ and set
    \[
        t=\floorbra{\eta\log\log n'},
    \]
    For sufficiently large $n'$, $\alpha<1/t$, so we may apply
    \Cref{thm:coloring-amplification} to $G$ with soundness parameter $\alpha$.
    Denote the amplified hypergraph by $G^{(t)}=(V^{(t)},E^{(t)})$. In the Yes case, $G^{(t)}$ is $(3^t,4^t)$-colorable. In the No case, every independent set in $G^{(t)}$ has size less than
    \[
        t\alpha |V^{(t)}|.
    \]

    \begin{samepage}
    Setting $c=4^t$, for sufficiently large $n'$, we have
    \[
        \varepsilon := t\alpha
        =t2^{-2^{\beta\sqrt{\log\log n'}}}
        \le 16^{-t} = \frac{1}{kc},
    \]
    where the inequality follows from
    $2^{\beta\sqrt{\log\log n'}}/\log\log n'\to\infty$.
    This is as required by \Cref{thm:main_reduction}.
    \end{samepage}
    The amplified hypergraph has size at most
    \[
        (n')^{O(4^t)}=(n')^{\mathrm{polylog}(n')},
    \]
    and therefore the whole reduction runs in time $n^{\mathrm{polylog}(n)}$. Let $N$ denote the size of the resulting UMPS instance. The gap produced by \Cref{thm:main_reduction} is
    \[
    \frac{ka+c(a+h)}{(k+2c)a}
    =\frac{2+(4/3)^t}{3}.
    \]
    Moreover, $N \le (n')^{O(4^t)}$, thus
    \[
            \log N\le O(4^t\log n')=(\log n')^{1+\eta\log 4+o(1)}
    \]
    while
    \[
        (4/3)^t=(\log n')^{\eta\log(4/3)+o(1)}.
    \]
    Hence, for some constant $\gamma>0$, the above gap is at least $(\log N)^\gamma$ for sufficiently large $n'$.

    A polynomial-time $(\log N)^\gamma$-approximation for unit-length UMPS would therefore solve $3\mathrm{SAT}$ in $n^{\mathrm{polylog}(n)}$ time, contradicting $\classnp\nsubseteq \dtime(n^{\mathrm{polylog}(n)})$.
\end{proof}

%% file: sec/coloring.tex
\section{Balanced Coloring Amplification}
\label{sec:hardness-coloring}
In this section, we prove \Cref{thm:coloring-amplification}, restated below.

\coloringamplification*

\subsection{Hypergraph Composition}

We use the following composition of hypergraphs, which dates back to \cite[Construction~(d)]{ErdosLovasz1975}.

\begin{definition}[Composition of Hypergraphs \cite{ErdosLovasz1975}]
\label{def:composition_product}
The composition of an $h_1$-uniform hypergraph $G_1=(V_1,E_1)$ and an $h_2$-uniform hypergraph $G_2=(V_2,E_2)$, denoted as $G_1\times G_2$, is the hypergraph $H=(V_H,E_H)$ defined by
\[
V_H=V_1\times V_2
\]
and
\[
E_H=\left\{
\bigcup_{u\in e}\bigl(\{u\}\times f_u\bigr)
:\ e\in E_1 \text{ and } f_u\in E_2 \text{ for every }u\in e
\right\}.
\]
\end{definition}

\begin{lemma}
\label{lem:composition_size}
Let $G_1=(V_1,E_1)$ be an $h_1$-uniform hypergraph and $G_2=(V_2,E_2)$ be an $h_2$-uniform hypergraph. Their composition $H=(V_H,E_H)$ satisfies
\[
|V_H|=|V_1|\cdot |V_2|,
\qquad
|E_H|=|E_1|\cdot |E_2|^{h_1},
\]
and $H$ is $h_1h_2$-uniform.
\end{lemma}
\begin{proof}
$|V_H|$ follows immediately from the definition. $|E_H|=|E_1|\cdot |E_2|^{h_1}$ because every hyperedge in $H$ is specified by a hyperedge $e \in E_1$ and $h_1$ many hyperedges in $E_2$.
Furthermore, each hyperedge in $H$ is a disjoint union of $h_1$ parts, each consisting of $h_2$ pairs, and thus has size $h_1h_2$.
\end{proof}

We next prove the completeness and soundness properties of the above composition in \Cref{lem:composition_completeness} and \Cref{lem:composition_soundness}, respectively.

\begin{lemma}
\label{lem:composition_completeness}
Suppose $G_1=(V_1,E_1)$ is an $(a_1,k_1)$-colorable hypergraph and $G_2=(V_2,E_2)$ is an $(a_2,k_2)$-colorable hypergraph. Then their composition $H=(V_H,E_H)$ is an $(a_1a_2,k_1k_2)$-colorable hypergraph.
\end{lemma}
\begin{proof}
Let $\chi_1$ be an $(a_1,k_1)$-coloring of $G_1$ and $\chi_2$ be an $(a_2,k_2)$-coloring of $G_2$. For a vertex $v=(v_1,v_2)\in V_H$, define
\[
\chi(v)=\phi(\chi_1(v_1),\chi_2(v_2)),
\]
where $\phi$ is an arbitrary bijection between $[k_1]\times[k_2]$ and $[k_1k_2]$.

Fix a hyperedge
\[
e=\bigcup_{u\in e_1}\bigl(\{u\}\times f_u\bigr)\in E_H,
\]
where $e_1\in E_1$ and $f_u\in E_2$ for every $u\in e_1$.
Fix a color $r\in[k_1k_2]$ and let $(r_1,r_2)=\phi^{-1}(r)$. Then
\begin{align*}
|e\cap\chi^{-1}(r)|
&=\sum_{u\in e_1\cap\chi_1^{-1}(r_1)}
    |f_u\cap\chi_2^{-1}(r_2)|\\
&\le a_2\,|e_1\cap\chi_1^{-1}(r_1)|\\
&\le a_1a_2. \qedhere
\end{align*}
\end{proof}

\begin{lemma}
\label{lem:composition_soundness}
Let $G_1=(V_1,E_1)$ and $G_2=(V_2,E_2)$ be two hypergraphs, and let $0<\varepsilon_1,\varepsilon_2<1$. Suppose every independent set in $G_i$ has size less than $\varepsilon_i|V_i|$ for $i\in\{1,2\}$. Then every independent set in their composition $H=(V_H,E_H)$ has size less than
\[
(\varepsilon_1+\varepsilon_2)|V_1||V_2|.
\]
\end{lemma}
\begin{proof}
Let $I\subseteq V_H$ be an independent set. For every $u\in V_1$, define the fiber
\[
I_u=\{v\in V_2:(u,v)\in I\},
\]
and let
\[
S=\{u\in V_1:I_u\text{ is not an independent set in }G_2\}.
\]
We claim that $S$ is an independent set in $G_1$. Suppose otherwise that there is some $e\in E_1$ with $e\subseteq S$. For every $u\in e$, the set $I_u$ is not independent in $G_2$, so we can choose a hyperedge $f_u\in E_2$ with $f_u\subseteq I_u$. By \Cref{def:composition_product},
\[
\bigcup_{u\in e}\bigl(\{u\}\times f_u\bigr)
\]
is a hyperedge of $H$ contained in $I$, a contradiction. Therefore, $|S|<\varepsilon_1|V_1|$.

For every $u\notin S$, the fiber $I_u$ is independent in $G_2$, and hence $|I_u|<\varepsilon_2|V_2|$. Thus,
\begin{align*}
|I|
&=\sum_{u\in V_1}|I_u|\\
&<|S||V_2|+|V_1|\varepsilon_2|V_2|\\
&<(\varepsilon_1+\varepsilon_2)|V_1||V_2| \ . \qedhere
\end{align*}
\end{proof}

\subsection{Proof of \Cref{thm:coloring-amplification}}

\begin{proof}[Proof of \Cref{thm:coloring-amplification}]
Construct $G^{(t)}$ as the $t$-th composition power of $G$, namely $G^{(1)}=G$ and $G^{(s)}=G\times G^{(s-1)}$ for $2\le s\le t$. Iterating \Cref{lem:composition_size} gives
\[
|V(G^{(t)})|=n^t,
\qquad
|E(G^{(t)})|=m^{1+h+\cdots+h^{t-1}},
\]
and $G^{(t)}$ is $h^t$-uniform. The construction runs deterministically in time polynomial in the graph size. %

For completeness, suppose $G$ is $(a,k)$-colorable. We prove by induction on $s$ that $G^{(s)}$ is $(a^s,k^s)$-colorable. The case $s=1$ is the assumption. If $s>1$, then $G$ is $(a,k)$-colorable and, by induction, $G^{(s-1)}$ is $(a^{s-1},k^{s-1})$-colorable. Applying \Cref{lem:composition_completeness} to $G\times G^{(s-1)}$ shows that $G^{(s)}$ is $(a^s,k^s)$-colorable.

For soundness, suppose every independent set in $G$ has size less than $\varepsilon |V|$. We prove by induction on $s$ that every independent set in $G^{(s)}$ has size less than $s\varepsilon |V(G^{(s)})|$. The case $s=1$ is the assumption. For $s>1$, apply \Cref{lem:composition_soundness} with $G_1=G$, $G_2=G^{(s-1)}$, $\varepsilon_1=\varepsilon$, and $\varepsilon_2=(s-1)\varepsilon$. This gives the threshold $s\varepsilon$, completing the proof.
\end{proof}

%% file: sec/pcp-free.tex
\section{A PCP-free $4/3$-Inapproximability}
\label{sec:pcp-free}

In this appendix, we give a direct reduction that yields the NP-inapproximability of unit-length UMPS. For every fixed integer $k\ge 1$, we obtain a gap
\[
A_k:=3k+1 \quad\text{vs.}\quad B_k:=4k+1,
\]
and therefore rule out any approximation factor less than $4/3$ in polynomial time. The proof is self-contained and does not involve the PCP theorem.

\begin{theorem}
\label{thm:pcp-free-gap-family}
For every fixed $k\ge 1$, there is a polynomial-time reduction from a 3SAT instance $\Phi$ to a unit-length UMPS instance  $I_k(\Phi)$, such that
\begin{itemize}
    \item (Yes) If $\Phi$ is satisfiable, then $\operatorname{OPT}(I_k(\Phi)) \le A_k$.
    \item (No) If $\Phi$ is not satisfiable, then $\operatorname{OPT}(I_k(\Phi)) \ge B_k$.
\end{itemize}
\end{theorem}

Before proceeding to the construction, we first define some notations used in this section. For a Boolean value $b\in\{\mathsf T,\mathsf F\}$, let $\bar b$ denote the other value.  For a literal $\ell$ over variable $x_i$, define its satisfying value by
\[
\operatorname{sat}(\ell)=
\begin{cases}
\mathsf T, & \ell=x_i,\\
\mathsf F, & \ell=\neg x_i.
\end{cases}
\]
If $I$ is a UMPS instance and $z$ is a job, then $I\prec z$ means that every job in $I$ precedes $z$, and $z\prec I$ is defined analogously.

We define the instances by induction.  Let $I_0$ be the one-job instance, which has optimum $1=A_0=B_0$. The 4 vs 5 instance $I_1(\Phi)$ is defined below. For every $k \ge 2$, the instance $I_k(\Phi)$ is constructed recursively using $I_0(\Phi),\ldots,I_{k-1}(\Phi)$, and will be specified later.

\paragraph{The 4-vs-5 Construction $I_1(\Phi)$.}
Let $\Phi$ be a $3\mathrm{SAT}$ instance with variables $x_1,\ldots,x_n$ and clauses $C_1,\ldots,C_m$.
For each variable $x_i$, create one machine $M_i$. Put on this machine two value jobs
\[
J_i^{\mathsf T},\quad J_i^{\mathsf F}.
\]
For each clause 
\[
C_j=(\ell_{j,1}\vee \ell_{j,2}\vee \ell_{j,3}),
\] create one machine $M_j$ with three jobs $P_{j,1},P_{j,2},P_{j,3}$.  If $u(j,t)$ is the variable index of $\ell_{j,t}$, add
\[
J_{u(j,t)}^{\operatorname{sat}(\ell_{j,t})}\prec P_{j,t}
\quad (t=1,2,3).
\]
This defines $I_1(\Phi)$. See \Cref{fig:pcp-free-base} for an illustration.

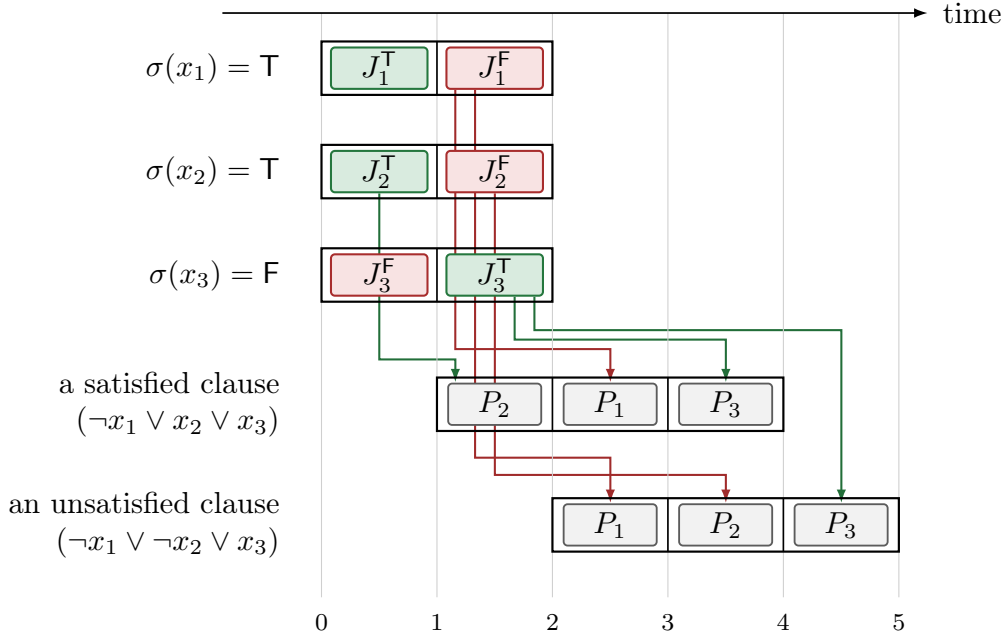
\begin{figure}[H]
\centering
\begin{tikzpicture}[x=1cm,y=1cm,scale=1.14,transform shape,>=Latex,font=\small]
  \pgfdeclarelayer{bg}
  \pgfsetlayers{bg,main}
  \def\slotW{1.34}
  \def\slotH{0.62}
  \def\xzero{4.85}
  \def\yA{4.15}
  \def\yB{2.95}
  \def\yC{1.75}
  \def\ySat{0.25}
  \def\yUnsat{-1.15}
  \tikzset{
    tjob/.style={draw=slotgreen!82!black,fill=slotgreen!18,line width=0.78pt,rounded corners=1.5pt,minimum width=1.12cm,minimum height=0.48cm,inner sep=1.2pt},
    fjob/.style={draw=slotred!82!black,fill=slotred!14,line width=0.78pt,rounded corners=1.5pt,minimum width=1.12cm,minimum height=0.48cm,inner sep=1.2pt},
    auxjob/.style={draw=gray!72!black,fill=gray!10,line width=0.70pt,rounded corners=1.5pt,minimum width=1.08cm,minimum height=0.48cm,inner sep=1.2pt},
    mbox/.style={line width=0.82pt},
    unsatbox/.style={line width=0.82pt,dash pattern=on 3.4pt off 2.4pt,line cap=round,line join=round},
    depT/.style={draw=slotgreen!76!black,line width=0.78pt,-{Latex[length=1.8mm,width=1.25mm]}},
    depF/.style={draw=slotred!78!black,line width=0.78pt,-{Latex[length=1.8mm,width=1.25mm]}}
  }

  \foreach \t in {0,1,2,3,4,5} {
    \pgfmathsetmacro{\xx}{\xzero+\t*\slotW}
    \draw[gray!38] (\xx,{-1.68})--(\xx,5.10);
    \node[anchor=north,font=\scriptsize] at (\xx,-1.74) {$\t$};
  }
  \draw[-{Latex[length=1.8mm,width=1.3mm]},line width=0.65pt]
    ({\xzero-1.15},5.10)--({\xzero+5.25*\slotW},5.10) node[right=1pt] {time};

  \node[anchor=east] at ({\xzero-0.35},{\yA+0.5*\slotH}) {$\sigma(x_1)=\mathsf T$};
  \node[anchor=east] at ({\xzero-0.35},{\yB+0.5*\slotH}) {$\sigma(x_2)=\mathsf T$};
  \node[anchor=east] at ({\xzero-0.35},{\yC+0.5*\slotH}) {$\sigma(x_3)=\mathsf F$};
  \node[anchor=east,align=right] at ({\xzero-0.35},{\ySat+0.5*\slotH})
    {a satisfied clause\\$(\neg x_1 \vee x_2 \vee x_3)$};
  \node[anchor=east,align=right] at ({\xzero-0.35},{\yUnsat+0.5*\slotH})
    {an unsatisfied clause\\$(\neg x_1 \vee \neg x_2 \vee x_3)$};

  \node[tjob]  (x1T) at ({\xzero+0.5*\slotW},{\yA+0.5*\slotH}) {$J_1^{\mathsf T}$};
  \node[fjob]  (x1F) at ({\xzero+1.5*\slotW},{\yA+0.5*\slotH}) {$J_1^{\mathsf F}$};
  \draw[mbox] (\xzero,\yA) rectangle ({\xzero+2*\slotW},{\yA+\slotH});

  \node[tjob]  (x2T) at ({\xzero+0.5*\slotW},{\yB+0.5*\slotH}) {$J_2^{\mathsf T}$};
  \node[fjob]  (x2F) at ({\xzero+1.5*\slotW},{\yB+0.5*\slotH}) {$J_2^{\mathsf F}$};
  \draw[mbox] (\xzero,\yB) rectangle ({\xzero+2*\slotW},{\yB+\slotH});

  \node[fjob]  (x3F) at ({\xzero+0.5*\slotW},{\yC+0.5*\slotH}) {$J_3^{\mathsf F}$};
  \node[tjob]  (x3T) at ({\xzero+1.5*\slotW},{\yC+0.5*\slotH}) {$J_3^{\mathsf T}$};
  \draw[mbox] (\xzero,\yC) rectangle ({\xzero+2*\slotW},{\yC+\slotH});

  \foreach \yy in {\yA,\yB,\yC} {
    \draw[line width=0.62pt] ({\xzero+\slotW},\yy)--({\xzero+\slotW},{\yy+\slotH});
  }

  \node[auxjob] (satP2) at ({\xzero+1.5*\slotW},{\ySat+0.5*\slotH}) {$P_2$};
  \node[auxjob] (satP1) at ({\xzero+2.5*\slotW},{\ySat+0.5*\slotH}) {$P_1$};
  \node[auxjob] (satP3) at ({\xzero+3.5*\slotW},{\ySat+0.5*\slotH}) {$P_3$};
  \draw[mbox] ({\xzero+1*\slotW},\ySat) rectangle ({\xzero+4*\slotW},{\ySat+\slotH});
  \foreach \t in {2,3} {
    \draw[line width=0.62pt] ({\xzero+\t*\slotW},\ySat)--({\xzero+\t*\slotW},{\ySat+\slotH});
  }

  \node[auxjob] (unsatP1) at ({\xzero+2.5*\slotW},{\yUnsat+0.5*\slotH}) {$P_1$};
  \node[auxjob] (unsatP2) at ({\xzero+3.5*\slotW},{\yUnsat+0.5*\slotH}) {$P_2$};
  \node[auxjob] (unsatP3) at ({\xzero+4.5*\slotW},{\yUnsat+0.5*\slotH}) {$P_3$};
  \draw[mbox] ({\xzero+2*\slotW},\yUnsat)
    rectangle ({\xzero+5*\slotW},{\yUnsat+\slotH});
  \foreach \t in {3,4} {
    \draw[line width=0.62pt]
      ({\xzero+\t*\slotW},\yUnsat)--({\xzero+\t*\slotW},{\yUnsat+\slotH});
  }

  \begin{pgfonlayer}{bg}
    \coordinate (colA1) at ($ (x1F.south west)!0.100!(x1F.south east) $);
    \coordinate (colB1) at ($ (x1F.south west)!0.300!(x1F.south east) $);
    \coordinate (colD2) at ($ (x2F.south west)!0.500!(x2F.south east) $);
    \coordinate (colE3) at ($ (x3T.south west)!0.700!(x3T.south east) $);
    \coordinate (colF3) at ($ (x3T.south west)!0.900!(x3T.south east) $);

    \coordinate (satP2In)   at (colA1 |- satP2.north);
    \coordinate (satP1In)   at ($ (satP1.north west)!0.50!(satP1.north east) $);
    \coordinate (satP3In)   at ($ (satP3.north west)!0.50!(satP3.north east) $);
    \coordinate (unsatP1In) at ($ (unsatP1.north west)!0.50!(unsatP1.north east) $);
    \coordinate (unsatP2In) at ($ (unsatP2.north west)!0.50!(unsatP2.north east) $);
    \coordinate (unsatP3In) at ($ (unsatP3.north west)!0.50!(unsatP3.north east) $);

    \coordinate (laneF) at (0,1.42);
    \coordinate (laneE) at (0,1.31);
    \coordinate (laneA) at (0,1.20);
    \coordinate (laneC) at (0,1.08);
    \coordinate (laneB) at (0,-0.06);
    \coordinate (laneD) at (0,-0.26);

    \draw[depT]
      (x2T.south)
      -- (x2T.south |- laneC)
      -- (satP2In |- laneC)
      -- (satP2In);

    \draw[depF]
      (colA1)
      -- (colA1 |- laneA)
      -- (satP1In |- laneA)
      -- (satP1In);

    \draw[depF]
      (colB1)
      -- (colB1 |- laneB)
      -- (unsatP1In |- laneB)
      -- (unsatP1In);

    \draw[depF]
      (colD2)
      -- (colD2 |- laneD)
      -- (unsatP2In |- laneD)
      -- (unsatP2In);

    \draw[depT]
      (colE3)
      -- (colE3 |- laneE)
      -- (satP3In |- laneE)
      -- (satP3In);

    \draw[depT]
      (colF3)
      -- (colF3 |- laneF)
      -- (unsatP3In |- laneF)
      -- (unsatP3In);
  \end{pgfonlayer}
\end{tikzpicture}
\caption{The $4$ vs.~$5$ construction, shown under the assignment $\sigma(x_1)=\mathsf T$, $\sigma(x_2)=\mathsf T$, and $\sigma(x_3)=\mathsf F$.  The upper clause $(\neg x_1 \vee x_2 \vee x_3)$ is satisfied, while the lower clause $(\neg x_1 \vee \neg x_2 \vee x_3)$ is unsatisfied.}
\label{fig:pcp-free-base}
\end{figure}

Fix an assignment $\sigma:\{x_1,\ldots,x_n\} \to \{\mathsf T, \mathsf F\}$ that satisfies $\Phi$. For every variable $x_i$, schedule $J_i^{\sigma(x_i)}$ at time $0$ and $J_i^{\overline{\sigma(x_i)}}$ at time $1$. In every clause, at least one literal is true.  That corresponding clause job is released at time $1$, while the other two clause jobs are released no later than time $2$. Hence the three clause jobs fit on $M_j$ in the slots $[1,2]$, $[2,3]$, and $[3,4]$.  Thus $\operatorname{OPT}(I_1(\Phi))\le A_1=4$.

Conversely, suppose that $\Phi$ is unsatisfiable and that there is a schedule of $I_1(\Phi)$ with makespan strictly less than $5$.  On each variable machine $M_i$, the two unit value jobs cannot both start before time $1$; therefore at least one of $J_i^{\mathsf T}$ and $J_i^{\mathsf F}$ starts at time at least $1$.  Call its value $b_i$, and set $\sigma(x_i):=\bar b_i$.  Since $\sigma$ does not satisfy $\Phi$, some clause $C_j$ is false under $\sigma$.  For every $t\in\{1,2,3\}$, the satisfying value of $\ell_{j,t}$ is exactly $b_{u(j,t)}$, so all three jobs $P_{j,1},P_{j,2},P_{j,3}$ are released at time at least $2$.  They would have to fit on one machine inside $[2,C_{\max})$, whose length is strictly less than $3$, which is  impossible.  Therefore $\operatorname{OPT}(I_1(\Phi))\ge B_1=5$.

\paragraph{Recursive Construction for $k \ge 2$.}

Fix $k \ge 2$ and assume $I_i(\Phi)$ has been defined for all $0 \le i < k$. In our construction of $I_k(\Phi)$, we will recursively use instances $I_i(\Phi), 0 \le i < k$, to anchor some dummy jobs in specific time slots.

Let $\Phi$ be a $3\mathrm{SAT}$ instance with variables $x_1,\dots,x_n$ and clauses $C_1,\ldots,C_m$.
For each variable $x_i$, create one variable machine $M_i$.  This machine has $k-1$ left layers and one terminal right layer. For each left layer $r\in [k-1]$, put on $M_i$ the three jobs
\[
J_{i,r}^{\mathsf T},\quad J_{i,r}^{\mathsf D},\quad J_{i,r}^{\mathsf F}.
\]
Also put on $M_i$ the three right-layer jobs
\[
J_{i,R}^{\mathsf T},\quad J_{i,R}^{\mathsf D},\quad J_{i,R}^{\mathsf F},
\]
where ${\mathsf T, \mathsf D, \mathsf F}$ represent True/Dummy/False, respectively. Thus, there are $3(k-1)+3=3k$ jobs in total on each variable machine. %

In the following, we add auxiliary jobs of kinds $H,Q,S,U,V$. Each job in each of those categories is fresh, and must be scheduled on its own fresh machine. These jobs enforce consistency between layers and anchor the dummy jobs. For each $1 \le r < k-1$, add two length-$3$ chains
\[
J_{i,r}^{\mathsf T}\prec H_{i,r,\mathsf T}^{1}\prec H_{i,r,\mathsf T}^{2}\prec J_{i,r+1}^{\mathsf T},
\]
\[
J_{i,r}^{\mathsf F}\prec H_{i,r,\mathsf F}^{1}\prec H_{i,r,\mathsf F}^{2}\prec J_{i,r+1}^{\mathsf F}.
\]
At the last layer, add two length-$4$ cross chains
\[
J_{i,k-1}^{\mathsf T}\prec Q_{i,\mathsf T}^{1}\prec Q_{i,\mathsf T}^{2}\prec Q_{i,\mathsf T}^{3}\prec J_{i,R}^{\mathsf F},
\]
\[
J_{i,k-1}^{\mathsf F}\prec Q_{i,\mathsf F}^{1}\prec Q_{i,\mathsf F}^{2}\prec Q_{i,\mathsf F}^{3}\prec J_{i,R}^{\mathsf T}.
\]

We next anchor the dummy jobs.  For each left layer $r\in [k-1]$, add
\[
I_{r-1}(\Phi)\prec J_{i,r}^{\mathsf D}\prec S_{i,r}\prec I_{k-r}(\Phi).
\]
For the right dummy job, add
\[
I_{k-1}(\Phi)\prec U_i\prec J_{i,R}^{\mathsf D}\prec V_i.
\]
Finally, for each clause $C_j$, create a machine $M_j$ containing three jobs $P_{j,1},P_{j,2},P_{j,3}$.  Let $u(j,t)$ be the variable index of literal $\ell_{j,t}$, and use cyclic notation $u(j,4)=u(j,1)$, $\ell_{j,4}=\ell_{j,1}$.  For each $t\in\{1,2,3\}$, add
\[
J_{u(j,t),k-1}^{\operatorname{sat}(\ell_{j,t})}
\prec P_{j,t}
\prec J_{u(j,t+1),R}^{\operatorname{sat}(\ell_{j,t+1})}.
\]
This completes the construction of $I_k(\Phi)$. For every $0 \le i < k$, the instance $I_i(\Phi)$ used to anchor the dummy jobs can be shared, so we only need one instance per starting time. Therefore, the reduction runs in time polynomial in $k$. See \Cref{fig:pcp-free-recursive-construction} for an illustration of the construction for $k=3$.

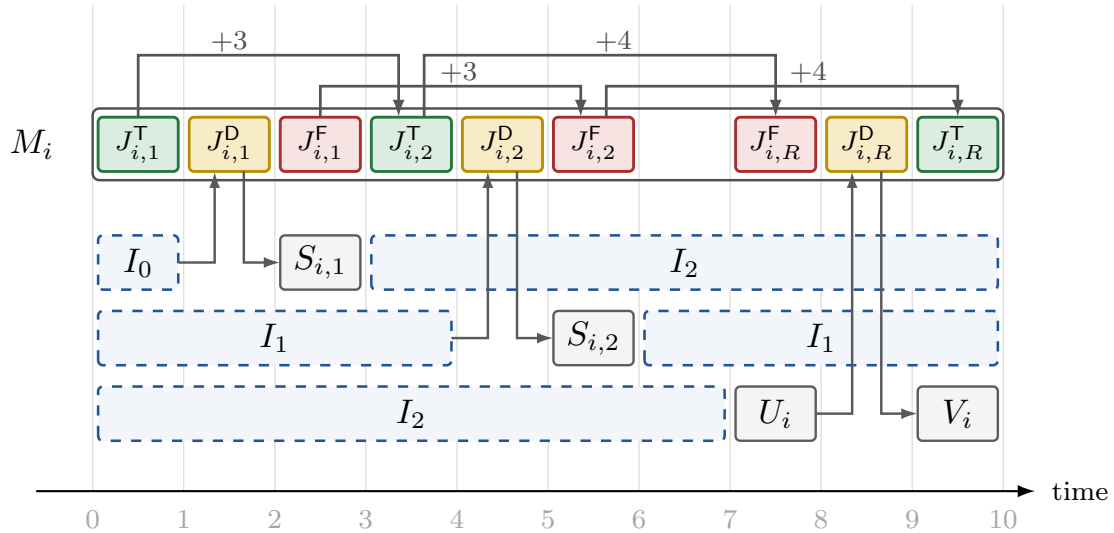
\begin{figure}[H]
\centering
\resizebox{0.9\textwidth}{!}{%
\begin{tikzpicture}[x=0.94cm,y=1cm,>=Latex,font=\small]
  \tikzset{
    truebox/.style={draw=slotgreen!82!black,fill=slotgreen!17,line width=0.84pt,rounded corners=1.4pt},
    dummybox/.style={draw=slotyellow!85!black,fill=slotyellow!22,line width=0.84pt,rounded corners=1.4pt},
    falsebox/.style={draw=slotred!82!black,fill=slotred!14,line width=0.84pt,rounded corners=1.4pt},
    machine/.style={draw=gray!62!black,line width=0.72pt,rounded corners=2pt},
    auxbox/.style={draw=gray!72!black,fill=gray!9,line width=0.70pt,rounded corners=1.4pt},
    recbox/.style={draw=slotblue!80!black,fill=slotblue!6,dashed,line width=0.76pt,rounded corners=1.4pt},
    chainarc/.style={draw=gray!67!black,line width=0.80pt,-{Latex[length=1.8mm,width=1.25mm]}},
    anchorlink/.style={draw=gray!70!black,line width=0.72pt,-{Latex[length=1.7mm,width=1.15mm]}},
    chainlabel/.style={inner sep=0.8pt,font=\scriptsize,text=gray!62!black},
    joblabel/.style={font=\scriptsize}
  }
  \def\jobY{0.22}
  \def\jobH{0.56}
  \def\jobM{0.50}
  \def\ancH{0.56}
  \def\ancM{0.28}
  \def\rOne{-1.00}
  \def\rTwo{-1.78}
  \def\rRight{-2.56}

  \foreach \t in {0,1,...,10} {
    \draw[gray!25] (\t,-3.08)--(\t,1.94);
    \node[anchor=north,font=\scriptsize,text=gray!70] at (\t,-3.14) {$\t$};
  }
  \draw[-{Latex[length=1.8mm,width=1.3mm]},line width=0.62pt]
    (-0.62,-3.08)--(10.35,-3.08) node[right=1pt,font=\scriptsize] {time};
  \draw[machine] (0,\jobY-0.09) rectangle (10,\jobY+\jobH+0.09);
  \node[anchor=east] at (-0.28,\jobY+\jobM*\jobH) {$M_i$};

  \draw[truebox] (0.06,\jobY) rectangle (0.94,\jobY+\jobH);
  \node[joblabel] at (0.5,\jobY+\jobM*\jobH) {$J_{i,1}^{\mathsf T}$};
  \draw[dummybox] (1.06,\jobY) rectangle (1.94,\jobY+\jobH);
  \node[joblabel] at (1.5,\jobY+\jobM*\jobH) {$J_{i,1}^{\mathsf D}$};
  \draw[falsebox] (2.06,\jobY) rectangle (2.94,\jobY+\jobH);
  \node[joblabel] at (2.5,\jobY+\jobM*\jobH) {$J_{i,1}^{\mathsf F}$};

  \draw[truebox] (3.06,\jobY) rectangle (3.94,\jobY+\jobH);
  \node[joblabel] at (3.5,\jobY+\jobM*\jobH) {$J_{i,2}^{\mathsf T}$};
  \draw[dummybox] (4.06,\jobY) rectangle (4.94,\jobY+\jobH);
  \node[joblabel] at (4.5,\jobY+\jobM*\jobH) {$J_{i,2}^{\mathsf D}$};
  \draw[falsebox] (5.06,\jobY) rectangle (5.94,\jobY+\jobH);
  \node[joblabel] at (5.5,\jobY+\jobM*\jobH) {$J_{i,2}^{\mathsf F}$};

  \draw[falsebox] (7.06,\jobY) rectangle (7.94,\jobY+\jobH);
  \node[joblabel] at (7.5,\jobY+\jobM*\jobH) {$J_{i,R}^{\mathsf F}$};
  \draw[dummybox] (8.06,\jobY) rectangle (8.94,\jobY+\jobH);
  \node[joblabel] at (8.5,\jobY+\jobM*\jobH) {$J_{i,R}^{\mathsf D}$};
  \draw[truebox] (9.06,\jobY) rectangle (9.94,\jobY+\jobH);
  \node[joblabel] at (9.5,\jobY+\jobM*\jobH) {$J_{i,R}^{\mathsf T}$};

  \def\tLane{1.42}
  \def\fLane{1.10}
  \draw[chainarc] (0.50,\jobY+\jobH) -- (0.50,\tLane) -- (3.36,\tLane) -- (3.36,\jobY+\jobH);
  \draw[chainarc] (3.64,\jobY+\jobH) -- (3.64,\tLane) -- (7.50,\tLane) -- (7.50,\jobY+\jobH);
  \draw[chainarc] (2.50,\jobY+\jobH) -- (2.50,\fLane) -- (5.36,\fLane) -- (5.36,\jobY+\jobH);
  \draw[chainarc] (5.64,\jobY+\jobH) -- (5.64,\fLane) -- (9.50,\fLane) -- (9.50,\jobY+\jobH);
  \node[chainlabel] at (1.50,\tLane+0.14) {$+3$};
  \node[chainlabel] at (5.72,\tLane+0.14) {$+4$};
  \node[chainlabel] at (4.02,\fLane+0.14) {$+3$};
  \node[chainlabel] at (7.86,\fLane+0.14) {$+4$};

  \draw[recbox] (0.06,\rOne) rectangle (0.94,\rOne+\ancH);
  \node at (0.5,\rOne+\ancM) {$I_0$};
  \draw[auxbox] (2.06,\rOne) rectangle (2.94,\rOne+\ancH);
  \node at (2.5,\rOne+\ancM) {$S_{i,1}$};
  \draw[recbox] (3.06,\rOne) rectangle (9.94,\rOne+\ancH);
  \node at (6.5,\rOne+\ancM) {$I_2$};
  \draw[anchorlink] (0.94,\rOne+\ancM)--(1.34,\rOne+\ancM)--(1.34,\jobY);
  \draw[anchorlink] (1.66,\jobY)--(1.66,\rOne+\ancM)--(2.06,\rOne+\ancM);

  \draw[recbox] (0.06,\rTwo) rectangle (3.94,\rTwo+\ancH);
  \node at (2,\rTwo+\ancM) {$I_1$};
  \draw[auxbox] (5.06,\rTwo) rectangle (5.94,\rTwo+\ancH);
  \node at (5.5,\rTwo+\ancM) {$S_{i,2}$};
  \draw[recbox] (6.06,\rTwo) rectangle (9.94,\rTwo+\ancH);
  \node at (8,\rTwo+\ancM) {$I_1$};
  \draw[anchorlink] (3.94,\rTwo+\ancM)--(4.34,\rTwo+\ancM)--(4.34,\jobY);
  \draw[anchorlink] (4.66,\jobY)--(4.66,\rTwo+\ancM)--(5.06,\rTwo+\ancM);

  \draw[recbox] (0.06,\rRight) rectangle (6.94,\rRight+\ancH);
  \node at (3.5,\rRight+\ancM) {$I_2$};
  \draw[auxbox] (7.06,\rRight) rectangle (7.94,\rRight+\ancH);
  \node at (7.5,\rRight+\ancM) {$U_i$};
  \draw[auxbox] (9.06,\rRight) rectangle (9.94,\rRight+\ancH);
  \node at (9.5,\rRight+\ancM) {$V_i$};
  \draw[anchorlink] (7.94,\rRight+\ancM)--(8.34,\rRight+\ancM)--(8.34,\jobY);
  \draw[anchorlink] (8.66,\jobY)--(8.66,\rRight+\ancM)--(9.06,\rRight+\ancM);
\end{tikzpicture}%
}
\caption{Illustration of the construction for $k=3$. The top row shows the single variable machine $M_i$ with $\sigma(x_i)=\mathsf T$ in the completeness schedule, and the arrows labeled $+3$ and $+4$ abbreviate the auxiliary chains. The lower rows show instances with lengths $A_0=1$, $A_1=4$, and $A_2=7$ to anchor the dummy jobs.}
\label{fig:pcp-free-recursive-construction}
\end{figure}

\paragraph{Completeness.}
Assume that $\Phi$ is satisfiable, and fix a satisfying assignment $\sigma:\{x_1,\ldots,x_n\}\to\{\mathsf T,\mathsf F\}$.  By the induction hypothesis, every recursive copy $I_i(\Phi)$ with $0 \le i < k$ has a schedule of length $A_i=3i+1$.

For each variable $x_i$, and each left layer $r\in [k-1]$, schedule on $M_i$
\[
J_{i,r}^{\sigma(x_i)}\text{ at time }3(r-1),
\quad
J_{i,r}^{\mathsf D}\text{ at time }3(r-1)+1,
\quad
J_{i,r}^{\overline{\sigma(x_i)}}\text{ at time }3(r-1)+2.
\]
On the same machine, schedule the right-layer jobs as
\[
J_{i,R}^{\overline{\sigma(x_i)}}\text{ at time }3k-2,
\quad
J_{i,R}^{\mathsf D}\text{ at time }3k-1,
\quad
J_{i,R}^{\sigma(x_i)}\text{ at time }3k.
\]
These slots are disjoint; the only idle slot on $M_i$ is $[3k-3,3k-2)$.
The $H$- and $Q$-chains fit exactly between their endpoints.  The left anchors fit because
\[
A_{r-1}+1+1+A_{k-r}=3k+1=A_k,
\]
and the right anchors fit because
\[
A_{k-1}+1+1+1=3k+1=A_k.
\]

It remains to schedule the clause machines.  For a clause $C_j$, let $\tau_t\in\{0,1\}$ indicate whether literal $\ell_{j,t}$ is true under $\sigma$, and recall that $\tau_4=\tau_1$.  Put $L:=3k-5$.  Based on our construction, the unit interval for $P_{j,t}$ can be placed inside the following time window, according to the pair $(\tau_t,\tau_{t+1})$:
\[
\begin{array}{c|c}
(\tau_t,\tau_{t+1}) & \text{time window containing } P_{j,t} \\\hline
(1,1) & [L,\,L+5),\\
(1,0) & [L,\,L+3),\\
(0,1) & [L+2,\,L+5),\\
(0,0) & [L+2,\,L+3).
\end{array}
\]
Since the clause is satisfied, its cyclic truth pattern is not $000$.  If it is $111$, all three jobs can be placed consecutively starting at time $L$.  If it has exactly two true literals, schedule the job with $(\tau_t,\tau_{t+1})=(1,0)$ at time $L$, the job with $(\tau_t,\tau_{t+1})=(1,1)$ at time $L+1$, and the job with $(\tau_t,\tau_{t+1})=(0,1)$ at time $L+2$.  If it has exactly one true literal, schedule the job with $(\tau_t,\tau_{t+1})=(1,0)$ at time $L$, the job with $(\tau_t,\tau_{t+1})=(0,0)$ at time $L+2$, and the job with $(\tau_t,\tau_{t+1})=(0,1)$ at time $L+3$.  Thus every clause machine is schedulable within the makespan $A_k$. %

\paragraph{Soundness.}
Now assume that $\Phi$ is unsatisfiable, and suppose for contradiction that $I_k(\Phi)$ has a schedule of makespan $C_{\max}<B_k=4k+1$.  By the induction hypothesis, every recursive copy $I_i(\Phi)$ with $0 \le i < k$ has makespan at least $B_i=4i+1$.

Let $s(z)$ denote the start time of a job $z$.  Consider a left dummy job $J_{i,r}^{\mathsf D}$.  Since it is preceded by a copy of $I_{r-1}(\Phi)$, we have
\[
s(J_{i,r}^{\mathsf D})\ge B_{r-1}=4r-3.
\]
Since it is followed by one extra job and then a copy of $I_{k-r}(\Phi)$, we also have
\[
s(J_{i,r}^{\mathsf D})+2+B_{k-r}<B_k,
\]
and therefore $s(J_{i,r}^{\mathsf D})<4r-2$.  Thus $J_{i,r}^{\mathsf D}$ occupies a unit interval contained in $[4r-3,4r-1)$.

For $r=1$, one of $J_{i,1}^\mathsf T$ and $J_{i,1}^\mathsf F$ must start at or after time $2$ on $M_i$: since $J_{i,1}^{\mathsf D}$ starts in $[1,2)$, at most one of these two unit jobs can start before time $2$ without overlapping it or the other value job.  Denote such a value by $b_i\in\{\mathsf T,\mathsf F\}$.  We claim that for every $r\in [k-1]$,
\begin{equation}
s(J_{i,r}^{b_i})\ge 4r-2.
\label{eq:pcp-free-late-left-3sat}
\end{equation}
This is true for $r=1$ by definition of $b_i$.  If it holds for some $r < k-1$, then the length-$3$ chain from $J_{i,r}^{b_i}$ to $J_{i,r+1}^{b_i}$ gives
\[
s(J_{i,r+1}^{b_i})\ge s(J_{i,r}^{b_i})+3 \ge 4r+1.
\]
The dummy job $J_{i,r+1}^{\mathsf D}$ starts in $[4r+1,4r+2)$, so a unit job on $M_i$ starting at or after $4r+1$ cannot be placed before it.  Hence $J_{i,r+1}^{b_i}$ must start after $J_{i,r+1}^{\mathsf D}$ completes, and therefore
\[
s(J_{i,r+1}^{b_i})\ge 4r+2 = 4(r+1)-2,
\]
proving \eqref{eq:pcp-free-late-left-3sat}.

Next consider the right dummy job.  Because of the anchor
\[
I_{k-1}(\Phi)\prec U_i\prec J_{i,R}^{\mathsf D}\prec V_i,
\]
we have
\[
s(J_{i,R}^{\mathsf D})\ge B_{k-1}+1=4k-2.
\]
Also, since $V_i$ must be scheduled after $J_{i,R}^{\mathsf D}$ and $C_{\max}<4k+1$, we have $s(J_{i,R}^{\mathsf D})<4k-1$.

Applying \eqref{eq:pcp-free-late-left-3sat} at the last left layer and then using the length-$4$ cross chain, we have
\[
s(J_{i,R}^{\bar b_i})\ge (4k-6)+4 = 4k-2.
\]
We now claim that
\begin{equation}
s(J_{i,R}^{b_i})<4k-2.
\label{eq:pcp-free-early-right-3sat}
\end{equation}
Indeed, if $s(J_{i,R}^{b_i})\ge 4k-2$, then the three jobs $J_{i,R}^{b_i}$, $J_{i,R}^{\bar b_i}$, and $J_{i,R}^{\mathsf D}$ would all have to be processed on $M_i$ during the interval $[4k-2,C_{\max})$.  This interval has length strictly less than $3$, so it cannot contain three unit-length jobs.

Define an assignment $\sigma$ by $\sigma(x_i):=\bar b_i$ for each variable $x_i$.  Since $\Phi$ is unsatisfiable, some clause $C_j$ is false under $\sigma$.  For each literal $\ell_{j,t}$ in that clause, the satisfying value of $\ell_{j,t}$ is exactly $b_{u(j,t)}$.  Hence, by \eqref{eq:pcp-free-late-left-3sat}, the predecessor of $P_{j,t}$ on the last left layer finishes no earlier than $4k-5$.  By \eqref{eq:pcp-free-early-right-3sat}, its successor on the right starts before $4k-2$.  Thus the three jobs $P_{j,1},P_{j,2},P_{j,3}$ must all be scheduled on one machine inside an interval of length strictly less than $3$, which is impossible.  Therefore $\operatorname{OPT}(I_k(\Phi))\ge B_k$.

\begin{corollary}
\label{cor:pcp-free-four-thirds}
For every constant $\varepsilon>0$, unit-length UMPS is NP-hard to approximate within a factor $4/3-\varepsilon$.
\end{corollary}

\begin{proof}
Let $\rho<4/3$ be fixed, and choose $k$ such that $\rho<(4k+1)/(3k+1)$.  By \Cref{thm:pcp-free-gap-family}, it is NP-hard to distinguish $\operatorname{OPT}\le 3k+1$ from $\operatorname{OPT}\ge 4k+1$.  A $\rho$-approximation would distinguish the two cases because $\rho(3k+1)<4k+1$.
\end{proof}